\documentclass[10pt]{article}

\usepackage{amsmath}
\usepackage{amssymb}
\usepackage{amsthm}
\usepackage[dvipsnames]{xcolor}
\usepackage{enumitem}
\usepackage{float}
\usepackage{hyperref}
\usepackage{mathtools}
\usepackage{mathrsfs}
\usepackage{microtype}
\usepackage{caption}
\usepackage{subcaption}
\usepackage{todonotes}
\usepackage{graphicx}
\graphicspath{ {./plots/} }

\usepackage[a4paper, total={6in, 9in}]{geometry}
\hypersetup{
colorlinks=true,
linkcolor=red,
filecolor=magenta,
urlcolor=cyan,
citecolor=green,
linktoc=all,
}

\newcommand{\E}{\mathbb{E}}
\newcommand{\F}{\mathbb{F}}

\newcommand{\N}{\mathbb{N}}
\renewcommand{\P}{\mathbb{P}}

\newcommand{\R}{\mathbb{R}}
\renewcommand{\S}{\mathbb{S}}

\newcommand{\V}{\mathbb{V}}

\newcommand{\cA}{\mathcal{A}}
\newcommand{\cB}{\mathcal{B}}
\newcommand{\cC}{\mathcal{C}}
\newcommand{\cD}{\mathcal{D}}

\newcommand{\cF}{\mathcal{F}}

\newcommand{\cH}{\mathcal{H}}

\newcommand{\cJ}{\mathcal{J}}

\newcommand{\cM}{\mathcal{M}}
\newcommand{\cN}{\mathcal{N}}

\newcommand{\cS}{\mathcal{S}}

\newcommand{\cV}{\mathcal{V}}

\newcommand{\cY}{\mathcal{Y}}
\newcommand{\cZ}{\mathcal{Z}}

\newcommand{\fN}{\mathfrak{N}}

\newcommand{\rD}{\mathrm{D}}
\newcommand{\rd}{\mathrm{d}}

\newcommand{\argdot}{\,\cdot\,}

\newcommand{\ind}{\mathbb{I}}

\DeclareMathOperator{\diag}{diag}

\DeclareMathOperator{\trace}{tr}

\theoremstyle{plain}
\newtheorem{theorem}{Theorem}[section]
\newtheorem{proposition}[theorem]{Proposition}
\newtheorem{corollary}[theorem]{Corollary}
\newtheorem{lemma}[theorem]{Lemma}
\newtheorem{definition}[theorem]{Definition}
\newtheorem{assumption}{Assumption}

\theoremstyle{definition}
\newtheorem{remark}[theorem]{Remark}

\newcommand{\LSC}{\mathrm{LSC}(\overline{\S}_T)}
\newcommand{\USC}{\mathrm{USC}(\overline{\S}_T)}
\newcommand{\LSCExit}{\mathrm{LSC}(\overline{\cS})}
\newcommand{\USCExit}{\mathrm{USC}(\overline{\cS})}
\newcommand{\fakeQED}{\hspace*{\fill}$\boldsymbol\square$}
\newcommand{\cont}{\mathbb{C}}
\newcommand{\interv}{\mathbb{D}}

\begin{document}
\title{Optimal Investment with Costly Expert Opinions}
\author{Christoph Knochenhauer \and Alexander Merkel \and Yufei Zhang}
\date{}

\maketitle

\begin{abstract}
We consider the Merton problem of optimizing expected power utility of terminal
wealth in the case of an unobservable Markov-modulated drift. What makes the
model special is that the agent is allowed to purchase costly expert opinions
of varying quality on the current state of the drift, leading to a mixed
stochastic control problem with regular and impulse controls involving random
consequences. Using ideas from filtering theory, we first embed the original
problem with unobservable drift into a full information problem on a larger state
space. The value function of the full information problem is characterized as
the unique viscosity solution of the dynamic programming PDE. This characterization
is achieved by a new variant of the stochastic Perron's method, which additionally
allows us to show that, in between purchases of expert opinions, the problem
reduces to an exit time control problem which is known to admit an optimal
feedback control. Under the assumption of sufficient regularity of this feedback
map, we are able to construct optimal trading and expert opinion strategies.
\end{abstract}

\section{Introduction}\label{sec:Introduction}

Over the last three decades, there has been tremendous interest in
continuous-time portfolio optimization problems with partial observations. These
problems are typically set up such that the agent is assumed to observe the
prices of financial assets whose returns have an unobservable drift. In the
Markovian case, these problems can be formulated as
hidden Markov models which, by using ideas from filtering theory, can be
transformed into full information problems at the cost of a higher-dimensional
state space; c.f.~\cite{elliott2008hidden}.
In fact, three special cases are known in which the transformed problem
remains finite-dimensional.
In the Bayesian case as e.g.\ considered in \cite{karatzas2001bayesian},
the hidden component of the drift is a static random variable of arbitrary
distribution. In contrast, the Kalmann--Bucy case as e.g.\ studied in
\cite{brendle2006portfolio} arises if the unobservable drift is taken to
be an Ornstein--Uhlenbeck process. Finally, there is the Wonham case as e.g.\ considered
in \cite{rieder2005portfolio} in which the drift is modulated by an unobservable
continuous-time finite-state Markov chain.
There are many extensions of these baseline models, as for example problems with
different objective functionals and more general financial market models including
jump diffusion models and non-Markovian dynamics; see for example \cite{honda2003optimal,
lakner1995utility, lakner1998optimal, putschogl2008optimal, sass2004optimizing} and
the references therein for a non-exhaustive list of the vast literature in this field.

In the present article, we consider an extension of the problem in which the
agent can purchase noisy expert opinions on the current state of the drift.
That is, an expert opinion at time $t$ is assumed to take the form
\begin{equation}\label{eq:expert-opinion-introduction}
  Z = q \mu_t + (1-q)\cN,
\end{equation}
where $\mu_t$ denotes the unobservable drift at time $t$, $\cN$ is an
independent random variable causing the expert opinion to be noisy, and
$q\in[0,1]$ is the quality of the expert opinion.
In general, the inclusion of expert opinions as above is not a new idea and
has been studied in \cite{frey2012portfolio,
frey2014portfolio, gabih2014expert, sass2017expert}. We furthermore mention
the papers \cite{sass2021diffusion, sass2023diffusion} which investigate
diffusion approximations of expert opinions as the time period in between the
arrival of new expert opinions tends to zero. Finally, there is an 
alternative approach to expert opinions which are revealed continuously in
time; see in particular \cite{davis2013black} and the related articles
\cite{davis2016simple, davis2021risk}.

What these articles have in common is that they assume that expert opinions
arrive periodically in time, that is, either at discrete times, Poisson times,
or continuously in time. In particular, expert opinions are treated as exogenous,
and the additional information they contain on the true state of the drift is
given to the agent free of charge. In our model,
on the other hand, we assume that the agent has to purchase expert opinions
in exchange for a strictly positive fee. To be precise, we consider a financial
market
model similar to \cite{rieder2005portfolio} in which there is a single risky
asset with drift modulated by a finite-state continuous-time Markov chain, i.e.\
we place ourselves in the Wonham setting. In addition to choosing a trading strategy,
the agent can decide on the timing of purchase and quality level of expert
opinions. Since each purchase is assumed
to incur a strictly positive cost and the agent's wealth is finite, an expert
opinion strategy is naturally modeled as an impulse control. This leads to the
formulation of the optimization problem as a mixed stochastic control problem
with partial observations involving both regular controls (the trading strategy)
and impulse controls (the expert opinion strategy). Moreover, due to the
noise $\cN$ in the expert opinion~\eqref{eq:expert-opinion-introduction}, impulse
controls have random consequences in the sense of~\cite{helmes2024modeling,
korn1997optimal}.

The precise model specifications and the formulation of the investment objective
can be found in Section~\ref{sec:The-Model}. Let us highlight here that special
care needs to be taken in the setup of the information structure as it may
be optimal to purchase multiple expert opinions at the same time instant, hence
requiring multiple simultaneous updates of the agent's information set.
In Section~\ref{sec:transformation-to-full-observations}, we apply standard
filtering results to transform the partial observation problem into a full
information problem at the cost of increasing the state space. This transformation
is a classical first step in the study of partial observation problems and
has the advantage of putting the dynamic programming machinery at our disposal.

The dynamic programming PDE of the full information problem, which takes the
form of Hamilton--Jacobi--Bellman quasi-variational inequalities (HJBQVI), is
studied in detail in Section~\ref{sec:viscosity-characterization}. To be precise,
we show that the HJBQVI admit a unique continuous viscosity solution denoted
by $V^+$. This is achieved by a variant of the stochastic Perron's method
which combines ideas which have previously appeared in
\cite{bayraktar2013stochastic, belak2022optimal}. Typically, the reason for
choosing the stochastic Perron's method is that it allows to characterize
the value function of a stochastic control problem as the unique viscosity
solution of the associated dynamic programming PDE without having to prove the
notoriously challenging dynamic programming principle first \cite{claisse2016pseudo}.
We, on the other hand, content ourselves (at this point) with the insight that
the unique viscosity solution $V^+$ dominates the value function of the full
information problem.

After having constructed the unique viscosity solution of the HJBQVI, we
proceed in Section~\ref{sec:exit-time-problem} by showing that, on the set $\cont$
of states on which purchases of expert opinions are strictly suboptimal,
the full information problem reduces to an exit time control problem. The
intuition here is that on $\cont$, the agent only chooses the trading strategy
and can ignore purchases of expert opinions up until the time of exit from $\cont$.
This connection is formalized by another version of the stochastic Perron's method.
More precisely, writing $V^\cont$ for the value function of the exit time
control problem, we first argue that the pointwise supremum $V^-$ of a suitable
set of stochastic subsolutions of the exit time control problem is a
viscosity supersolution of the associated dynamic programming PDE such that
$V^-\leq V^\cont$. On the other hand, it is straightforward to show that
the unique viscosity solution $V^+$ of the HJBQVI is a viscosity subsolution of the
dynamic programming PDE of the exit time control problem with $V^+ \geq V^\cont$.
A comparison principle for the PDE shows that $V^- = V^\cont = V^+$, thus
establishing the connection between the full information problem and the
exit time control problem.

The advantage of establishing this connection is that the exit time control
problem is known to admit an optimal control, which we can use to construct
an optimal trading strategy and optimal expert opinion strategy for the full
information problem. This construction is also carried out in
Section~\ref{sec:exit-time-problem}. A verification theorem furthermore shows
that the unique viscosity solution $V^+$ of the HJBQVI coincides with the
value function of the full information problem, thus completing the main
results of this paper.

The construction of
the optimal strategies, however, involves a number of subtle technicalities and
requires several novel ideas, hence making up a significant part of the
main contributions of this article.
First, as usual for partial observation problems, the associated
full information problem is highly degenerate in the sense that the $(1+N)$-dimensional
state process is driven by a one-dimensional Brownian motion. In particular,
due to the lack of uniform ellipticity, there is no hope to find a classical
solution of the HJBQVI, and hence classical verification-type arguments to
construct optimizers is out of reach. To work around this problem, we
employ an extension of the verification procedure developed in
\cite{belak2019utility, belak2017general, belak2022optimal, christensen2014solution}
which allows to construct optimal expert opinion strategies (or, more generally,
impulse controls) from continuous viscosity solutions. One central assumption
in this type of construction is that the underlying stochastic process
satisfies the strong Markov property. Unfortunately, in our situation, this
is in general not guaranteed to hold as trading strategies (the regular
controls) can in principle be path-dependent and lead to non-Markovian state
processes. In the literature, there is very little work on impulse control
of non-Markovian processes. Let us mention \cite{djehiche2010stochastic,
jonsson2023finite} in this direction, but highlight that none of the approaches
developed in these articles seems feasible in our setting.

For the resolution of the lack of Markovianity, we make particular use of the connection
to the exit time control problem. More precisely, according to the results in
\cite{haussmann1990existence}, our exit time control problem
admits an optimal control. Using this, we establish the martingale optimality principle
which shows that the composition of the value function of the exit time control
problem with the optimally controlled state process is a martingale. This martingale
property is sufficient to replace the strong Markov arguments in the impulse
verification machinery of \cite{belak2019utility, belak2017general,
belak2022optimal, christensen2014solution}.

Another challenge in the construction of the optimal strategies is that the
optimal control of the exit time control problem provided by
\cite{haussmann1990existence} is only known to exist in a weak formulation of
the control problem, that is, on a particular probability space with a
particular driving Brownian motion and particular filtration. We work around
this issue by assuming that the stochastic differential equation associated
with the optimal control admits a strong solution, hence allowing us to
construct the optimal state process on arbitrary probability spaces and
arbitrary driving Brownian motions. While we could drop this assumption by
posing our optimal investment problem in a weak formulation as well, it is
our impression that the additional technicalities involved in doing so would
rather distract from the main contributions of the paper. In that sense,
our assumption of existence of a strong solution is simply a means of
keeping the technical burden (and length) of this paper reasonable.

\section{Formulation of the Optimal Investment Problem}
\label{sec:The-Model}

For a given finite time horizon $T>0$, let $(\Omega,\cF,\F,\P)$ be a complete
filtered probability space such that $\F=(\cF_{s})_{s\in[0,T]}$ satisfies the
usual assumptions of completeness and right-continuity and write
$\cF_\infty:=\cF_T$. On this space, we
assume that we are given an adapted continuous-time Markov chain
$Y=(Y_{s})_{s\in[0,T]}$ with càdlàg paths. We assume that $Y$ takes values in
$\{e_{1},\dots,e_{N}\}$, the canonical unit vectors in $\R^N$, where $N\in\N$
denotes the number of regimes of the economy. We take as given an $\cF_0$-measurable
random variable $y_0$ with values in $\{e_{1},\dots,e_{N}\}$ and assume that $Y_0 = y_0$.
The distribution of $y_0$
is determined by the probability vector $p_0 = (p_0^1,\dots,p_0^N)\in\Delta^{N}$,
where
\[
  p^n_0 := \P[y_0 = e_n],\quad n=1,\dots,N,
  \qquad\text{and}\qquad
  \Delta^N := \Bigl\{p\in[0,1]^N : \sum_{n=1}^N p^n = 1\Bigr\}.
\]
Moreover, the generator matrix
of $Y$ is denoted by $Q\in\R^{N\times N}$, so that $Y$ can be written as
\[
  \rd Y_{s} = QY_{s}\rd s + \rd M_{s},\quad Y_{0}=y_{0} \sim p_{0},\quad s\in[0,T],
\]
where $M = (M_s)_{s\in[0,T]}$ is a càdlàg martingale satisfying $M_{s}=\E[Y_{s} - Y_{0}|\cF_{s}]$,
$s\in[0,T]$. Finally, we take as given a one-dimensional standard $(\F,\P)$-Brownian motion
$B$ independent of $Y$.

The financial market is assumed to consist of a risk-free asset with zero returns
and a single risky asset $S=(S_{s})_{s\in[0,T]}$ which satisfies
\begin{align*}
  \rd S_{s}&= S_{s}\bigl[\mu^{\intercal}Y_s\rd s + \sigma\rd B_{s}\bigr],
  \quad S_{0} = s_{0},\quad s\in[0,T],
\end{align*}
where $s_{0}>0$ is the deterministic initial price, $\mu\in\R^{N}$ are
the returns in the respective states of the economy, and $\sigma>0$ is the
volatility.

We assume below that the agent observes the risky asset $S$, but does not observe the
state of the economy $Y$ directly. Instead, the current state of the economy is
filtered from the risky asset price. In addition, we assume that the agent has
access to costly expert opinions of varying quality. As we will assume the cost
for an expert opinion to be strictly positive, the agent's expert opinion strategies
are naturally modeled as impulse controls.

\begin{definition}\label{def:observation-controls}
An expert opinion strategy $\nu=(\tau,q)$ consists of an increasing sequence
of $\F$-stopping times $\tau=(\tau_{k})_{k\in\N}$, the expert opinion times,
and an $(\cF_{\tau_{k}})_{k\in\N}$-adapted, $[0,1)$-valued process
$q=(q_{k})_{k\in\N}$, the quality levels.
The set of all expert opinion strategies is denoted by $\mathcal{A}^{exp}_{pre}$.
\end{definition}

Expert opinion strategies $\nu=(\tau,q)\in\mathcal{A}^{exp}_{pre}$ reveal
additional information on the state of the economy $Y$ as follows. Given $k\in\N$,
the $k$-th expert opinion $Z_k^{\nu}=Z_{k}^{\tau,q}$ obtained at time
$\tau_k$ with quality level $q_k$ is defined on $\{\tau_k < \infty\}$ by
\begin{equation}\label{eq:def-expert-opinion}
  Z_{k}^{\nu} := q_k \mu^\intercal Y_{\tau_{k}} + (1-q_k)\cN_{k},
\end{equation}
where $\cN = (\cN_{k})_{k\in\N}$, the noise in the expert opinions, is a sequence
of independent and identically distributed random variables, independent of $Y$
and $B$, and assumed absolutely continuous with respect to the Lebesgue measure
with continuous density $\phi:\R\to[0,\infty)$.
From \eqref{eq:def-expert-opinion}, we see that the choice of $q_k=0$ results
in purely noisy observation $\cN_k$ whereas for $q_k \uparrow 1$, effectively,
the true state of the economy is revealed to the agent.
To motivate the particular form of \eqref{eq:def-expert-opinion}, note that, in
terms of information, observing $Z^\nu_k$ with $q_k>0$ is equivalent to observing
\begin{equation}\label{eq:def-expert-opinion-alt}
  \frac{1}{q_k}Z_{k}^{\nu} = \mu^\intercal Y_{\tau_{k}} + \frac{1-q_k}{q_k}\cN_{k}
\end{equation}
since $q_k$ is chosen by the agent, hence observable. The quality level $q_k$ therefore controls the amount
of independent noise in the expert opinion. While
\eqref{eq:def-expert-opinion-alt} has a more
straightforward economic interpretation, the advantage of working with
\eqref{eq:def-expert-opinion} instead is that this representation is continuous
in $q_k=0$, leading overall to cleaner arguments below.

At any given point in time, the agent's observations consist of the risky asset price $S$
and the acquired expert opinions $Z^\nu$ up to that point. To define the observation
filtration rigorously, let us fix an expert opinion strategy
$\nu=(\tau,q)\in\cA^{exp}_{pre}$.
With this, we first introduce the filtrations
\[
  \cY^{S}_{s} := \sigma\bigl(S_{r} : r\in[0,s]\bigr)
  \quad\text{and}\quad
  \cZ^{\nu,k}_s := \sigma\bigl((\ind_{\{\tau_k\leq r\}},Z_k^\nu\ind_{\{\tau_k\leq r\}}) : r\in[0,s]\bigr)
  \text{ for all }k\in\N,
  \quad s\in[0,T],
\]
where $\ind_{A}$ denotes the indicator function of an event $A$. The observation
filtration $\cY^\nu = (\cY^\nu_s)_{s\in[0,T]}$ associated with an expert opinion
strategy $\nu$ is then defined as
\[
  \cY^\nu_s = \cY^{S}_{s} \bigvee_{k\in\N} \cZ^{\nu,k}_s \bigvee \fN,
  \quad s\in[0,T],
\]
where $\fN$ denotes the system of $\P$-nullsets. Similarly, the observation filtration
restricted to the first $k\in\N_0$ expert opinions is denoted by $\cY^{\nu;k}
= (\cY^{\nu,k}_s)_{s\in[0,T]}$ and given by
\[
  \cY^{\nu,k}_s = \cY^{S}_{s} \bigvee_{j=1,\dots,k} \cZ^{\nu,j}_s \bigvee \fN,
  \quad s\in[0,T].
\]

\begin{remark}
Special care needs to be taken as multiple expert opinions may be bought at the
same time instant. Both our definition of expert opinion strategies and of the information
filtration account for this possibility.
\end{remark}

Trading is modeled in terms of fractions of wealth invested in the risky asset.
To be precise, a trading strategy is an $\F$-progressively measurable process $\pi$ taking values
in the interval $\Pi:=[\underline{\pi},\overline{\pi}]\subset\R$ with boundaries
$\underline{\pi}<\overline{\pi}$ and such that $0\in\Pi$. Given
$\nu=(\tau,q)\in\cA^{exp}_{pre}$, we write
$\tau_\infty := \lim_{k\to\infty} \tau_k$ for the accumulation point of
$\tau$. The wealth process $W^{w;u}=W^u=(W^u_s)_{s\in[0,\tau_\infty)\cap[0,T]}$ associated with $u:=(\pi,\nu)$
is defined as the solution of
\begin{equation}\label{eq:WEE}
  \rd W^{u}_{s}
  = \pi_{s} W^{u}_{s} \bigl[\mu^{\intercal}Y_{s}\rd s + \sigma\rd B_{s}\bigr]
	-\sum_{k\in\N} K(\tau_{k},q_{k}) \delta_{\tau_{k}}(\rd s),
	\quad W^u_{0}=w,
	\quad s\in[0,\tau_\infty)\cap[0,T],
\end{equation}
with $w\geq 0$ the initial wealth and the function $K:[0,T]\times[0,1)\to(0,\infty)$
modeling the cost of expert opinion acquisition. The cost function $K$ is assumed
to be jointly continuous and strictly increasing in the second argument. Moreover,
we assume that there exists $K_{\min}>0$ such that $K(t,q)\geq K(t,0) \geq K_{\min}$
for all $(t,q)\in[0,T]\times[0,1)$ and that $K(t,\argdot)$ is coercive in that
for each wealth level $w\geq 0$ there exists a quality level $q_w\in[0,1)$ such that
\begin{equation}\label{eq:coercive}
  K(t,q) \geq w,\quad(t,q)\in[0,T]\times[q_w,1).
\end{equation}
Observe that this assumption implies that $\lim_{q\uparrow 1}K(t,q) = \infty$ for
all $t\in[0,T]$, meaning that revealing the true state of the economy is assumed to
be infinitely costly.

\begin{definition}
\label{def:admissible-controls}
Let $u=(\pi,\nu)$, where $\pi$ is a trading strategy and $\nu=(\tau,q)\in\cA^{exp}_{pre}$.
We say that $u$ is an admissible strategy for an initial wealth of $w\geq 0$
provided that, for all $k\in\N$,
$\pi$ is $\cY^{\nu}$-progressively measurable,
$\tau_{k}$ is a $\cY^{\nu,k-1}$-stopping time,
$q_{k}$ is $\cY^{\nu,k-1}_{\tau_{k}}$-measurable, and
$W^{w;u}\geq 0$.
The set of all such strategies is denoted by $\cA(w)$. Given $u=(\pi,\nu)\in\cA(w)$,
we also speak of $\pi$ as an admissible trading strategy and $\nu$
as an admissible expert opinion strategy. 
\end{definition}

We shall shortly see that the admissibility constraint $W^{w;u}\geq 0$
guarantees that $\tau$ does not accumulate before time $T$. Intuitively, this is
because the agent starts with a finite amount of wealth and each expert opinion
costs at least $K_{\min}>0$. Since the financial market is free of arbitrage
opportunities, this implies that only finitely many expert opinions can be
financed.

Before making this argument rigorous, it is convenient to fix some notation first.
Given $u\in\cA(w)$, we subsequently always assume that the trading strategy in $u$
is denoted by $\pi$ and the expert opinion strategy is denoted by $\nu=(\tau,q)$.
Moreover, with a slight abuse of notation, we write $(\pi,\tau,q)$ in place of
$(\pi,(\tau,q))$.
While our notation carefully disentangles expert opinions even if they are
purchased simultaneously, it is sometimes more convenient to combine such simultaneous
purchases. This can be achieved as follows. Given $\nu=(\tau,q)\in\cA^{exp}_{pre}$,
we let $\tilde{\tau}=(\tilde{\tau}_k)_{{k\in\N}}$ denote the jump times of the càdlàg process
$I = (I_s)_{s\in[0,\infty)}$ given by
\[
  I_s := \sum_{k\in\N} \ind_{\{\tau_k \leq s\}},\quad s\in[0,\infty).
\]
Clearly, any jump in $I$ is due to at least one expert opinion and the accumulation
time $\tau_\infty$ of $\tau$ coincides with the first time that $I$ takes
the value $+\infty$. Moreover, while $\tau$ is in general increasing, the sequence
$\tilde{\tau}$ is even strictly increasing on $[0,\infty)$. Given $k\in\N$, the set of indices of
the expert opinions purchased at time $\tilde{\tau}_k$ is given by
\[
  J_k := \bigl\{ j\in\N : \tau_j = \tilde{\tau}_k \bigr\}.
\]
Note that any such set of indices can be written in the form $J_k
= \{j_{k,1},j_{k,2},\dots,j_{k,|J_k|}\}$ if $|J_k|$ is finite or $J_k
= \{j_{k,1},j_{k,2},\dots\}$ otherwise, where by convention we assume that
$j_{k,1}<j_{k,2}<\dots$ are ordered. In fact, it holds that
$j_{k,i+1} = j_{k,1} + i$ for all $i=0,1,\dots,|J_k|-1$ and
$(J_k)_{k\in\N}$ is a (random) partition of $\N$. Put differently, $J_{k}$ contains
the indices of the $k$-th batch of expert opinions and the $i$-th expert opinion
of that batch is the $j_{k,i}$-th expert opinion in total. For ease of notation,
we subsequently write
\[
  \tau_{k,i} := \tau_{j_{k,i}},
  \quad
  q_{k,i} := q_{j_{k,i}},
  \quad
  Z^\nu_{k,i} := Z^\nu_{j_{k,i}},
  \quad
  \cN_{k,i} := \cN_{j_{k,i}},
  \quad
  \cY^{\nu,k,i}:=\cY^{\nu,j_{k,i}}.
\]
Given a trading strategy $\pi$, the wealth process $W^{u}$ associated with
$u=(\pi,\tau,q)$ can be written as
\begin{equation}\label{eq:WEE-rewritten}
  \rd W^{u}_{s}
  =\pi_{s} W^{u}_{s} \bigl[\mu^{\intercal}Y_{s}\rd s + \sigma\rd B_{s}\bigr]
    - \sum_{k\in\N}\sum_{i=1}^{|J_{k}|} K(\tilde{\tau}_{k},q_{k,i})
    \delta_{\tilde{\tau}_{k}}(\rd s),
    \quad W^u_{0}=w_{0},
    \quad s\in[0,\tau_\infty)\cap[0,T].
\end{equation}
Finally, the wealth after the $i$-th expert opinion of the $k$-th batch is denoted by
\begin{equation}\label{eq:wealth-update-at-observation}
  W^{u}_{k,i} := W^{u}_{k,i-1} - K(\tilde{\tau}_{k},q_{k,i}),
  \quad W^{u}_{k,0} := W^{u}_{\tilde{\tau}_{k}-}
  \quad\text{on }\{\tilde{\tau}_k\leq T\}
\end{equation}
under the convention $W^u_{0-} := w_0$.
With the notation settled, let us now proceed by showing that admissible strategies
purchase only finitely many expert opinions.

\begin{lemma}
\label{lem:stop-to-inf}
For any $w\geq 0$ and $u=(\pi,\tau,q)\in\cA(w)$ we have $\tau_\infty>T$.
\end{lemma}

\begin{proof}
Let $u\in\cA(w)$ and define $A_{k,i}:=\{\tau_{k,i}\leq T\}$ for $i=1,\dots,|J_{k}|$ and $k\in\N$.
Moreover, let
\[
  A := \Bigl\{\lim_{k\to\infty}\tau_{k}\leq T\Bigr\}
  = \bigcap_{k\in\N}\bigcap_{i=1,\dots,|J_{k}|}A_{k,i}
\]
and assume by contradiction that $\P[A]>0$. Recall that for $k\in\N$ and $i=1,\dots,|J_{k}|$
the $i$-th expert opinion in the $k$-th batch corresponds to the $j_{k,i}$-th expert opinion
overall. Since the cost for expert opinions is lower bounded by $K_{\min}>0$, we see that
\[
  0 \leq W_{k,i}^{u}
  \leq w + \int_{0}^{\tau_{k,i}}\pi_{s}W_{s}\bigl[\mu^{\intercal}Y_{s}\rd s
    + \sigma\rd B_{s}\bigr] - j_{k,i}K_{\min}\quad \text{ on }A_{k,i}.
\]
In particular, this implies
\[
  j_{k,i}K_{\min} - w
  \leq \int_{0}^{\tau_{k,i}}\pi_{s}W_{s}\bigl[\mu^{\intercal}Y_{s}\rd s
    + \sigma\rd B_{s}\bigr]
  \quad\text{on }A.
\]
Now the financial market satisfies ``No Free Lunch with Vanishing Risk'', see
\cite[Theorem 8.2.1]{delbaen2006mathematics}, and hence the sequence of integrals
\[
  \int_{0}^{\tau_{k,i}\wedge T}\pi_{s}W_{s}\bigl[\mu^{\intercal}Y_{s}\rd s
    + \sigma\rd B_{s}\bigr],\quad k\in\N,
\]
is bounded in probability; see \cite[Lemma 8.2.4]{delbaen2006mathematics}. On the
other hand, $j_{k,i}K_{\min} - w\to\infty$ on $A$ as $k\to\infty$, from which
we obtain a contradiction to $\P[A]>0$.
\end{proof}

The previous lemma justifies modeling expert opinion purchases as impulse controls.
Moreover, it follows that for any admissible strategy $u\in\cA(w)$, the wealth process
$W^u$ is well-defined on all of $[0,T]$. As such, the optimization problem in which
the agent maximizes expected utility of terminal wealth over all admissible strategies
is well-posed. More precisely, given a constant relative risk aversion parameter
$\alpha\in(0,1)$, we denote by
\[
  U : [0,\infty)\to[0,\infty),\qquad w\mapsto U(w) := \frac{1}{1-\alpha}w^{1-\alpha},
\]
the associated power utility function,
so that the expected utility of terminal wealth
is given by
\begin{equation}\label{eq:cost-functional-partial}
	\cJ(w;u) := \E\bigl[U\bigl(W^{w;u}_T\bigr)\bigr]
  \quad\text{subject to \eqref{eq:WEE}},
  \qquad w\geq 0,\; u\in\cA(w).
\end{equation}
With this, the optimization problem considered in this article can be written as
\begin{equation}\label{eq:control-problem-partial}\tag{PO}
  \sup_{u\in\cA(w)} \cJ(w;u),\qquad w\geq 0.
\end{equation}
Formally, this is a stochastic control problem with partial observations featuring
both a ``classical'' absolutely continuous control $\pi$ and an ``impulse'' control
$(\tau,q)$. Moreover, due to the noise $\cN$ in the expert opinions, the impulse
controls have random consequences as in \cite{helmes2024modeling, korn1997optimal}.

\section{Transformation to Full Information}
\label{sec:transformation-to-full-observations}

As usual for stochastic control problems with partial observations, the first step
is to use filtering techniques to embed the original optimization
problem~\eqref{eq:control-problem-partial} into another stochastic control problem
with full information at the cost of an enlarged state space. The necessity for
such an embedding arises from the
fact that $Y$ is not adapted to the observation filtration $\cY^\nu$,
hence the drift $\pi W^u \mu^{\intercal} Y$ of the wealth process $W^u$ given
in~\eqref{eq:WEE} is, in general, also not adapted to $\cY^\nu$. As such, classical
solution techniques based on dynamic programming are not directly applicable and
it is more convenient to transform the problem to one with full information.

The first step in the transformation is the characterization of the conditional distribution
of the hidden Markov chain $Y$ given the observations $\cY^\nu$ for any strategy
$u=(\pi,\nu)=(\pi,\tau,q)\in\cA(w)$. For this, we denote by $p^{\nu,n}$ the
optional projection of $\ind_{\{Y = e_n\}}$ onto $\cY^\nu$ for $n=1,\dots,N$.
In particular, it holds that
\[
  p^{\nu,n}_s = \P\bigl[Y_{s}=e_{n}\big|\cY^{\nu}_{s}\bigr],
  \quad n=1,\dots,N,\; s\in[0,T],
\]
and the vector-valued process $p^{\nu}=(p^{\nu,1},\dots,p^{\nu,N})^\intercal$
is $\cY^\nu$-optional and takes values in $\Delta^N$.
In line with \eqref{eq:wealth-update-at-observation}, we furthermore define
\begin{equation}\label{eq:probability-update-at-observation}
  p^{\nu,n}_{k,i}:=\P\bigl[Y_{\tau_{k,i}}=e_{n}\big|\cY^{\nu,k,i}_{\tilde{\tau}_{k}}\bigr],
  \quad
  p^{\nu,n}_{k,0}:=p^{\nu,n}_{\tilde{\tau}_{k}-},
  \quad k\in\N,\; i=1,\dots,|J_{k}|,\; n=1,\dots,N,
\end{equation}
again with the convention $p^{\nu}_{0-}:=p_{0}$, where we recall that $p_0$ is
the initial (unconditional) distribution of $Y_0$. Next, we introduce the innovations
process $I^\nu = (I^\nu_s)_{s\in[0,T]}$ as
\begin{equation}\label{eq:def-innovations-process}
  \rd I^{\nu}_{s} := \frac{1}{\sigma}\bigl[\rd R_{s}-\mu^\intercal p^{\nu}_{s}\rd s\bigr]
  =\rd B_{s} + \frac{1}{\sigma}\mu^\intercal (Y_{s}-p^{\nu}_{s})\rd s,
  \quad I_{0}^{\nu}=0, \quad s\in[0,T],
\end{equation}
where $R := \log(S)$ denotes the return process, which, by It\^{o}'s formula and
\eqref{eq:def-innovations-process}, satisfies
\[
  \rd R_{s} = \Bigl[\mu^{\intercal}Y_{s}-\frac{1}{2}\sigma^{2}\Bigr]\rd s + \sigma\rd B_{s}
  = \Bigl[\mu^\intercal p^{\nu}_{s}-\frac{1}{2}\sigma^{2}\Bigr]\rd s + \sigma\rd I^{\nu}_{s},
  \quad R_{0}=\log(s_{0}),
  \quad s\in[0,T].
\]
A classical result in filtering theory guarantees that for any $u\in\cA(w)$, the innovations
process $I^\nu$ is a standard $(\cY^\nu,\P)$-Brownian motion; see e.g.\
\cite[Lemma 22.1.7]{cohen2015stochastic} for the case in which the observation
filtration is generated by a diffusion. Observe that, here, we do make use of the fact
that the price process $S$ and the return process $R$ generate the same filtration;
see \cite[Section 2]{frey2012portfolio}.

We can now rewrite the dynamics of the wealth process using the innovations process
as driving noise, which results in the drift of the wealth process to be adapted
to the observation filtration. More precisely, for any $u=(\pi,\nu)\in\cA(w)$,
the wealth process $W^u$ can be written as
\begin{equation}\label{eq:wealth-with-IP}
  \rd W_{s}^{u}
  = \pi_{s} W^{u}_{s}\bigl[\mu^\intercal p_{s}^{\nu} \rd s+\sigma \rd I^{\nu}_{s}\bigr]
    - \sum_{k\in\N}\sum_{i=1}^{|J_{k}|} K(\tilde{\tau}_{k},q_{k,i})\delta_{\tilde{\tau}_{k}}(\rd s),
  \quad W_{0}^{u}=w,\quad s\in[0,T].
\end{equation}
From this representation of the wealth process, we immediately see that the drift
$\pi W^{u}\mu^\intercal p^{\nu}$ is $\cY^\nu$-adapted. However, we are still missing
the dynamics of the conditional probability vector $p^{\nu}$ to get a full description
of the state. As in \cite[Theorem 9.1]{liptser2013statisticsI},
the dynamic evolution of $p^\nu$ in between the arrival of new expert opinions is given by
\[
  \rd p^\nu_s = Q^\intercal p^\nu_s\rd s
    + \frac{1}{\sigma}\bigl(\diag[\mu] - \mu^\intercal p^\nu_s\bigr)p^\nu_s \rd I^\nu_s,
  \quad p^{\nu}_{\tilde{\tau}_{k}}=p^{\nu}_{k,|J_{k}|},
  \quad s\in[\tilde{\tau}_{k},\tilde{\tau}_{k+1})\cap[0,T],
\]
for all $k\in\N$, where $\diag[\mu]\in\R^{N\times N}$ is the diagonal matrix
with diagonal entries given by $\mu$.
Moreover, for $k\in\N$ and $i = 1,\dots,|J_k|$, the
arrival of the new expert opinion
\[
  Z_{k,i}^{\nu} = q_{k,i}\mu^\intercal Y_{\tilde{\tau}_{k}} + (1-q_{k,i})\cN_{k,i}
\]
leads to an update of the conditional probabilities, which can
be computed using Bayes' rule, that is
\begin{align*}
  p^{\nu,n}_{k,i}
  &= \P\bigl[Y_{\tilde{\tau}_{k}}=e_{n}\,\big|\,\cY^{\nu,k,i}_{\tilde{\tau}_{k}}\bigr]\\
  &= \P\bigl[Y_{\tilde{\tau}_{k}}=e_{n}\,\big|\,\cY^{\nu,k,i-1}_{\tilde{\tau}_{k}} \vee Z_{k,i}^{\nu}\bigr]\\
  &= \frac{\P\bigl[Z_{k,i}^{\nu}\in\rd z\,\big|\,\cY^{\nu,k,i-1}_{\tilde{\tau}_{k}}
      \vee Y_{\tilde{\tau}_{k}}=e_{n}\bigr]
	\P\bigl[Y_{\tilde{\tau}_{k}}=e_{n}\,\big|\,\cY^{\nu,k,i-1}_{\tilde{\tau}_{k}}\bigr]}
    {\P\bigl[Z_{k,i}^{\nu}\in \rd z\,\big|\,\cY^{\nu,k,i-1}_{\tilde{\tau}_{k}}\bigr]}\\
  &= \frac{\P\bigl[Z_{k,i}^{\nu}\in \rd z\,\big|\,\cY^{\nu,k,i-1}_{\tilde{\tau}_{k}}
      \vee Y_{\tilde{\tau}_{k}}=e_{n}\bigr]
    \P\bigl[Y_{\tilde{\tau}_{k}}=e_{n} \,\big|\, \cY^{\nu,k,i-1}_{\tilde{\tau}_{k}}\bigr]}
	{\sum_{m=1}^{N}\P\bigl[Z^{\nu}_{k,i}\in \rd z, Y_{\tilde{\tau}_{k}}=e_{m}\,
	  \big|\,\cY^{\nu,k,i-1}_{\tilde{\tau}_{k}}\bigr]}\\
  &= \frac{\phi\bigl((Z_{k,i}^{\nu}-q_{k,i}\mu^{n})/(1-q_{k,i})\bigr)p_{k,i-1}^{\nu,n}}
    {\sum_{m=1}^{N}\phi\bigl((Z_{k,i}^{\nu}-q_{k,i}\mu^{m})/(1-q_{k,i})\bigr)p_{k,i-1}^{\nu,m}}.
\end{align*}
In light of this representation, we define the update function for the conditional
probabilities as $\Xi:\R\times\Delta^{N}\times[0,1)\to\Delta^{N}$
component-wise for each $n=1,\dots,N$ by
\[
  \Xi^n(z,p,q)
  := \frac{\phi\bigl((z-q\mu^n)/(1-q)\bigr)p^{n}}
    {\sum_{m=1}^{N}\phi\bigl((z-q\mu^{m})/(1-q)\bigr)p^{m}},
  \quad (z,p,q)\in\R\times\Delta^{N}\times[0,1).
\]
Observe that $\Xi$ is continuous by continuity of $\phi$.
Next, the update can equivalently be expressed as
\[
  \Xi^n(z,p,q)-p^{n}
  = \zeta^n(z,p,q)p^n,
\]
where the continuous function $\zeta:\R\times\Delta^N\times[0,1)\to\R^N$ is defined component-wise by
\begin{equation}\label{eq:def-zeta-update}
  \zeta^n(z,p,q) := \frac{\phi\bigl((z-q\mu^n)/(1-q)\bigr)}
  {\sum_{m=1}^{N}\phi\bigl((z-q\mu^{m})/(1-q)\bigr)p^{m}}-1,
  \quad (z,p,q)\in\times\R\times\Delta^N\times[0,1),
  \; n=1,\dots,N.
\end{equation}

\begin{remark}\label{rem:tractable-cases}
The conditional distribution of $Y$ given $\cY^\nu$ can be described by a finite-dimensional
system of equations since $Y$ takes only finitely many values. Other cases in which
the filter is finite-dimensional are the Bayesian case of $Y$ being a static random variable
as for example in \cite{karatzas2001bayesian} or the Gaussian case of $Y$ being an
Ornstein--Uhlenbeck process as for example in \cite{brendle2006portfolio}.
\end{remark}

Having derived the dynamics of $p^\nu$ essentially completes the transformation
of the partial observation problem to an equivalent full information problem.
However, in order to apply the tools of dynamic programming, we still have to
embed the full information problem into an entire family of optimization problems
parametrized in terms of the initial time $t\in[0,T]$, initial wealth $w\in[0,\infty)$
at time $t$, and initial conditional distribution $p\in\Delta^N$ of $Y$ at time $t$.
This embedding is standard and outlined here merely for the sake of introducing
the proper notation.

We begin by setting
\[
  \S := [0,\infty)\times\Delta^N,
  \quad
  \S_T := [0,T)\times \S,
  \quad\text{and}\quad
  \overline{\S}_T := [0,T]\times\S
\]
for the state space, the time-augmented state space, and the time-augmented state
space including terminal time, respectively. Elements of $\S$ are denoted by $x=(w,p)
=(w,p^1,\dots,p^N)$ and we occasionally switch between these representations
without explicitly mentioning this.
The drift of the
state process is given by $f = (f^w,f^p) : \S\times\Pi \to \R\times\R^N$ with
\begin{equation}\label{eq:def-drift-coefficients}
  f^{w}(w,p,\pi) := \pi w \mu^\intercal p
  \quad\text{and}\quad
  f^{p}(p) := Q^\intercal p,
  \quad (w,p,\pi)\in\S\times\Pi.
\end{equation}
Similarly, the diffusion coefficient
$\Sigma = (\sigma^{w},\sigma^{p}) : \S\times\Pi \to \R\times\R^N$ is defined through
\begin{equation}\label{eq:def-diffusion-coefficients}
  \sigma^{w}(w,\pi) := \pi w\sigma
  \quad\text{and}\quad
  \sigma^{p}(p) := \frac{1}{\sigma}(\diag(\mu) - \mu^{\intercal}p)p.
  \quad (w,p,\pi)\in\S\times\Pi.
\end{equation}
Jumps in the state process due to the purchase of new expert opinions are given in
terms of the continuous function
$\gamma=(\gamma^w,\gamma^p):\R\times[0,T]\times\Delta^N\times[0,1)\to\R\times\R^N$ given by
\[
  \gamma^w(z,t,p,q) := - K(t,q)
  \quad\text{and}\quad
  \gamma^p(z,t,p,q) = \zeta(z,p,q)\diag[p]
\]
for $(z,t,p,q)\in \R\times[0,T]\times\Delta^N\times[0,1)$.
For any $(t,x)=(t,w,p)\in\overline{\S}_T$, we subsequently write $\cA(x):=\cA(w)$
and given an
admissible strategy $u\in\cA(x)$, the $\S$-valued state process
$X^{u;t,x}=(X^{u;t,x}_s)_{s\in[t,T]}$ is given as the unique strong solution of
\begin{equation}\label{eq:state-aux}
  \rd X^{u;t,x}_s = f\bigl(X^{u;t,x}_s,\pi_s\bigr)\rd s
    + \Sigma\bigl(X^{u;t,x}_s,\pi_s\bigr)\rd I^\nu_s
	+ \sum_{k\in\N}\sum_{i=1}^{|J_k|}
	  \gamma\bigl(Z^\nu_{k,i},\tilde{\tau}_k,p^\nu_{k,i-1},q_{k,i}\bigr)
	  \delta_{\tilde{\tau}_k}(\rd s)
\end{equation}
for $s\in[t,T]$ with initial condition $X^{u;t,x}_t = x$, and
where we impose the convention $p^{\nu}_{k,0}:=p^{\nu}_{\tilde{\tau}_{k}-}$.
In line with the notation introduced above,
we often write $X^u$ in place of $X^{u;t,x}$ if the initial condition is clear from
the context, we generally denote the components of $X^u$ by $X^u=(W^u,p^\nu)$, and
for $i=1,\dots,|J_{k}|$ and $k\in\N$ we write
$X_{k,i}^{u}:=(W_{k,i}^{u},p_{k,i}^{\nu})$ for the state process after the
$i$-th expert opinion of the $k$-th batch.

\begin{lemma}
\label{lem:quad-growth}
There exists $C>0$ such that
\[
  \E\Bigl[\sup_{s\in[t,T]}\bigl|X^{u;t,x}_{s}\bigr|^{2}\Bigr]
  \leq C\bigl(1+|x|^{2}\bigr)
\]
for any $(t,x)\in\overline{\S}_T$ and any admissible strategy $u\in\cA(x)$.
\end{lemma}

\begin{proof}
Writing $X^u = X^{u;t,x}$ and $X^u=(W^u,p^\nu)$, we first make use of the fact
that all entries of $p^\nu$ are $[0,1]$-valued to estimate
\[
  \bigl|X^{u}_{s}\bigr|^{2}
  \leq (N+1)\Bigl[\bigl|W^u_s\bigr|^2 + \sum_{n=1}^N \bigl|p^{\nu,n}\bigr|^2\Bigr]
  \leq (N+1)\bigl|W^u_s\bigr|^2 + N(N+1).
\]
Next, we note that $W^u \leq \hat{W}^{u}$, where $\hat{W}^u$ corresponds to the wealth
process $W^{u}$ in the absence of expert opinion cost, that is
\[
  \rd\hat{W}_{s}^{u} 
  = f^{w}\bigl(\hat{W}^{u}_{s}, p^{\nu}_{s}, \pi_{s}\bigr)\rd s
	+ \sigma^{w}\bigl(\hat{W}^{u}_{s}, \pi_{s}\bigr)\rd I^{\nu}_{s},
  \qquad \hat{W}^{u}_{t} = w,
  \qquad s\in[t,T].
\]
Since $\Pi$ and $\Delta^N$ are bounded, we conclude that $\hat{W}$ solves a
Lipschitz SDE with Lipschitz constant independent of
$u$. But then a standard moment estimate such as
\cite[Corollary~2.2.12]{krylov2008controlled} shows that there exists a constant
$K>0$ which does not depend on $u$ and $t$ such that
\[
  \E\Bigl[\sup_{s\in[t,T]}\bigl|X^{u}_{s}\bigr|^{2}\Bigr]
  \leq N(N+1) + (N+1)\E\Bigl[\sup_{s\in[t,T]}\bigl|\hat{W}^{u}_{s}\bigr|^{2}\Bigr]
  \leq N(N+1) + (N+1)K\bigl(1 + |w|^2\bigr).
\]
We conclude since $|w| \leq |x|$.
\end{proof}

With this, we can now formulate the dynamic full information problem as follows.
The expected utility of terminal wealth for a given initial condition
$(t,x)\in\overline{\S}_T$ and admissible control $u\in\cA(x)$ is
\[
  \cJ(u;t,x) := \E\bigl[U\bigl(W^{u;tx}_T\bigr)\bigr]
  \quad \text{subject to \eqref{eq:state-aux},}
\]
and the value function, in turn, is defined as
\begin{equation}\label{eq:VF-full}\tag{FI}
  V : \overline{\S}_T \to[0,\infty),
  \quad (t,x)\mapsto V(t,x) := \sup_{u\in\cA(x)} \cJ(u;t,x).
\end{equation}

\section{Construction of the Viscosity Solution of the HJBQVI}
\label{sec:viscosity-characterization}

In this section we begin with our characterization of the value function as the
unique viscosity solution of the Hamilton--Jacobi--Bellman quasi-variational
inequalities (HJBQVI) associated with the full information problem. For this,
we employ a version of the stochastic Perron's method; see
\cite{bayraktar2012stochastic,bayraktar2013stochastic} for early developments
and \cite{belak2017general} for impulse control problems.
Similarly to \cite{belak2019utility, belak2017general,
belak2022optimal}, we choose a suitable class of superharmonic functions for the
set of stochastic supersolutions.
A classical argument is employed to show that the pointwise minimum
of the set of stochastic supersolution is a viscosity subsolution of the HJBQVI.
In contrast to the classical stochastic Perron's method for impulse control problems
\cite{belak2019utility,belak2017general}, however, we then follow a similar argument
as in \cite{belak2022optimal} and show directly that the pointwise minimum is in
fact the unique continuous viscosity solution of the HJBQVI. As most of the
technical arguments in this section are well-known from the literature, we defer the
proofs to Appendix~\ref{app:perron}.

Note that, in order to show that the HJBQVI characterizes the value function,
after this section we are still left with showing that the pointwise minimum of
stochastic supersolutions coincides with the value function. This is achieved by
means of a novel argument proving that the full information problem can be
reduced to an exit time control problem. Once this is achieved, we use a
verification theorem to establish the viscosity characterization of the value
function $V$ and construct optimal strategies in the process.

\subsection{The HJBQVI}

To introduce the HJBQVI, we write $\Gamma:\R\times\overline{\S}_T\times[0,1)\to\S$
for the continuous function
\begin{equation}\label{eq:def-Gamma}
  \Gamma(z,t,x,q) := x + \gamma(z,t,p,q),
  \quad (z,t,x,q)\in\R\times\overline{\S}_T\times[0,1),
\end{equation}
which maps the current state $x$ to the new state $\Gamma(z,t,x,q)$ after
the purchase of an expert opinion $z$ of quality $q$ at time $t$. Observe
that since the expert opinion cost $K$ is bounded away from zero, it is possible
that the agent does not have sufficient wealth to buy any expert opinion at all.
In light of this, we subsequently introduce the set-valued mapping
\[
  \cD(t,x) := \bigl\{ q \in [0,1) : w \geq K(t,q)\bigr\},
  \quad (t,x)=(t,w,p)\in\overline{\S}_T.
\]
With this, the subset of the state space $\overline{\S}_T$ on which the agent cannot
purchase any expert opinions is given by
\[
  \overline{\S}_T^\emptyset
  := \bigl\{(t,x)\in\overline{\S}_T : \cD(t,x) = \emptyset\bigr\}
  = \bigl\{(t,w,p)\in\overline{\S}_T : w < K(t,0)\bigr\}.
\]
Observe that the complement is the closed set given by
\[
  \overline{\S}_T\setminus\overline{\S}_T^\emptyset
  = \bigl\{(t,w,p)\in\overline{\S}_T : w \geq K(t,0)\bigr\}.
\]
Next, note that monotonicity and continuity of $K$
imply that there exists a continuous function $\chi:\overline{\S}_T\to\R$
such that
\[
  \cD(t,x) = [0,\chi(t,x)],\quad (t,x)\in\overline{\S}_T.
\]
In particular, by monotonicity of $K$, it holds that $\chi < 0$ on
$\overline{\S}_T^\emptyset$ and $\chi$ is determined implicitly by
\[
  K\bigl(t,\chi(t,x)\bigr) = w,
  \quad (t,x)=(t,w,p)\in\overline{\S}_T\setminus\overline{\S}_T^\emptyset.
\]
Finally, note that $\chi<1$ since $K$ is assumed coercive.

Now fix $(t,x)=(t,w,p)\in\overline{\S}_T$. By definition, the optimal expected
utility the agent can obtain by starting from this state is $V(t,x)$. If the
agent decides to purchase an expert opinion of the form $Z = q\mu^\intercal Y_t
+ (1-q)\cN$ for a quality level $q\in[0,1)$ and a realization of noise $\cN\sim\phi$,
the expected optimal expected utility after this purchase is
\[
  \E\bigl[V\bigl(t,\Gamma(Z,t,x,q)\bigr)\bigr]
    = \sum_{n=1}^N p^n\E\Bigl[V\Bigl(t,\Gamma\bigl(q\mu^n + (1-q)\cN,t,x,q\bigr)\Bigr)\Bigr].
\]
In general, such a purchase is suboptimal, suggesting that
\[
  V(t,x) \geq
    \sum_{n=1}^N p^n\E\Bigl[V\Bigl(t,\Gamma\bigl(q\mu^n + (1-q)\cN,t,x,q\bigr)\Bigr)\Bigr],
\]
with equality only if the purchase of $Z$ is optimal. This leads to the definition
of the intervention operator\footnote{We use the convention
$\sup\{\emptyset\} = -1$ to handle the case $\cD(t,x) = \emptyset$.}
\begin{equation}\label{eq:def-impulse-operator}
  \cM[v](t,x) := \sup_{q\in\cD(t,x)}\sum_{n=1}^N p^n
    \E\Bigl[v\Bigl(t,\Gamma\bigl(q\mu^n+(1-q)\cN,t,x,q\bigr)\Bigr)\Bigr],
  \quad (t,x)\in\overline{\S}_T,
\end{equation}
acting on measurable non-negative functions $v:\overline{\S}_T\to[0,\infty)$. In
light of the discussion above, we expect that $V(t,x)\geq\cM[V](t,x)$ with equality
only if a purchase of an expert opinion in state $(t,x)$ constitutes an optimal
action. In particular, note that $\cM[V] = -1 < 0 \leq V$ on
$\overline{\S}_T^\emptyset$.

Next, we denote by $\cH:\S\times\Pi\times\R^{1+N}\times\cS^{1+N}\to\R$ the
Hamiltonian of the state process given by
\[
  \cH(x,\pi,r,M) := f(x,\pi)^\intercal r
    + \frac{1}{2}\trace\bigl[\Sigma(x,\pi)\Sigma(x,\pi)^\intercal M\bigr],
  \quad (x,\pi,r,M)\in \S\times\Pi\times\R^{1+N}\times\cS^{1+N},
\]
where $\cS^{1+N}$ denotes the set of symmetric matrices in $\R^{(1+N)\times(1+N)}$.
Now writing
\[
  F(x,s,r,M) := - s - \max_{\pi\in\Pi} \cH(x,\pi,r,M),
  \quad (x,s,r,M)\in\S\times\R\times\R^{1+N}\times\cS^{1+N},
\]
where $F:\S\times\R\times\R^{1+N}\times\cS^{1+N}\to\R$, the Hamilton--Jacobi--Bellman
quasi-variational inequalities can be expressed as
\begin{equation}\label{eq:HJBQVI}\tag{HJBQVI}
  0 = \min\Bigl\{F\bigl(x,V_t(t,x),\rD_xV(t,x),\rD_x^2V(t,x)\bigr),
    V(t,x) - \cM[V](t,x)\Bigr\},
  \quad (t,x)\in\S_T,
\end{equation}
with terminal condition
\[
  V(T,x) = U(w),\quad x=(w,p)\in\S.
\]
In what follows, we are concerned with viscosity solutions of \eqref{eq:HJBQVI}
which satisfy the terminal condition in the classical sense. As is common in the
theory of viscosity solutions, we denote by $v^*$ and $v_*$ the
upper and lower semi-continuous envelopes of a locally bounded function
$v:\overline{\S}_T\to\R$. Moreover, we write $\LSC$ and $\USC$ for the sets of
lower and upper semi-continuous functions $v:\overline{\S}_T\to\R$. We use the
standard definition of viscosity solutions for discontinuous HJBQVI as in, e.g.,
\cite{belak2022optimal}. For the sake of completeness and to keep the paper
self-contained, we recall this definition.

\begin{definition}
A locally bounded function $v:\overline{\S}_T\to\R$ is a viscosity subsolution
of \eqref{eq:HJBQVI} if, for any $(t,x)\in\S_T$ and any $\varphi\in\cC^{1,2}(\S_T)$
such that $v^{*}-\varphi$ attains a global maximum at $(t,x)$ with
$v^{*}(t,x)=\varphi(t,x)$, we have
\[
  \min\Bigl\{F\bigl(x,\varphi_t(t,x),\rD_x\varphi(t,x),\rD_x^2\varphi(t,x)\bigr),
    v^*(t,x) - \cM[v^*]^*(t,x)\Bigr\}
  \leq 0.
\]
Similarly, $v$ is a viscosity supersolution of \eqref{eq:HJBQVI} if, for any
$(t,x)\in\S_T$ and any $\varphi\in\cC^{1,2}(\S_T)$ such that $v_{*}-\varphi$ attains
a global minimum at $(t,x)$ with $v_{*}(t,x)=\varphi(t,x)$, we have
\[
  \min\Bigl\{F\bigl(x,\varphi_t(t,x),\rD_x\varphi(t,x),\rD_x^2\varphi(t,x)\bigr),
    v_{*}(t,x) - \cM[v_*]_*(t,x)\Bigr\}
  \geq 0.
\]
Finally, $v$ is a viscosity solution of \eqref{eq:HJBQVI} if it is both a viscosity
subsolution and a viscosity supersolution of \eqref{eq:HJBQVI}.
\end{definition}

The main reason why we repeat the definition here is to explicitly point out the,
at first sight, rather odd double use of the semi-continuous envelopes in
$\cM[v^*]^*$ and $\cM[v_*]_*$, respectively. This is required as the intervention
operator does, in general, not preserve semi-continuity. More precisely, as the
following result shows, $\cM$ preserves upper semi-continuity so that
$\cM[v^*]^*=\cM[v^*]$, but the same is generally not true for the lower
semi-continuous envelope.

\begin{lemma}\label{lem:regularity-intervention}
Let $v:\overline{\S}_T\to[0,\infty)$. Then the following statements hold.
\begin{enumerate}[label=\roman*),ref=\roman*)]
\item If $v\in\USC$, then $\cM[v]\in\USC$.
\item If $v\in\LSC$, then the restrictions of $\cM[v]$ to $\overline{\S}_T^\emptyset$
and $\overline{\S}_T\setminus\overline{\S}_T^\emptyset$ are lower semi-continuous
on their respective domains.
\end{enumerate}
\end{lemma}

The proof of this result is standard and follows along the same lines as
\cite[Lemma 5.1]{belak2022optimal}. There is a minor additional challenge since
the set of quality levels $[0,1)$ is not compact, which can be overcome using
that the cost function $K$ is coercive. The proof is reported in Appendix~\ref{app:perron}.

\subsection{The Comparison Principle}

The first step in the construction of the unique viscosity solution of the HJBQVI
is a comparison principle which we shall eventually use to prove uniqueness and
continuity of the viscosity solution. The arguments follow the standard machinery
as in \cite{belak2019utility,belak2022optimal} by first constructing a strict
supersolution of the HJBQVI and then proving the comparison principle via
a perturbation method.

\begin{lemma}
\label{lem:strict-classical-super}
For any choice of $\beta\in(0,\alpha]$, there exist constants $A,C>0$
and a strictly positive continuous function $\kappa:\overline{\S}_T\to(0,\infty)$ such that
\[
  \psi : \overline{\S}_T\to(0,\infty),
  \quad (t,x)\mapsto \psi(t,x) := \Bigl(A + \frac{1}{1-\beta}w^{1-\beta}\Bigr)e^{C(T-t)}
\]
is a strict supersolution of the HJBQVI in the sense that $\psi(T,w,p)\geq U(w)$
for all $(p,w)\in\S$ and
\[
  \min\Bigl\{F\bigl(x,\psi_t(t,x),\rD_x\psi(t,x),\rD_x^2\psi(t,x)\bigr),
  \psi(t,x) - \cM[\psi](t,x)\Bigr\}
  \geq \kappa(t,x) > 0,
  \quad (t,x)\in\S_T.
\]
\end{lemma}

\begin{remark}\label{rem:stoch_supersolution_member}
The proof of Lemma~\ref{lem:strict-classical-super} reveals that
the function $\psi$ with $\beta = \alpha$ and $A=0$ is supersolution provided that
$C>0$ is chosen sufficiently large. It is, however, in general no longer a strict
supersolution.
\end{remark}

With the strict supersolution at hand, we can derive the following perturbation
result which goes back to \cite[Lemma 3.2]{ishii1993viscosity}.

\begin{lemma}\label{lem:perturbation}
Fix $\rho>1$, let $u\in\USC,v\in\LSC$, and $\psi$, $\kappa$ as in
Lemma~\ref{lem:strict-classical-super}. Define the perturbations
\begin{equation}\label{eq:def-perturb-fcts}
  u^{\rho} := \frac{\rho+1}{\rho}u - \frac{1}{\rho}\psi,
  \quad
  v^{\rho} := \frac{\rho-1}{\rho}v+\frac{1}{\rho}\psi.
\end{equation}
If $u$ is a viscosity subsolution of \eqref{eq:HJBQVI}, the perturbation $u^{\rho}$
is a viscosity subsolution of the perturbed equation
\[
  \min\Bigl\{F\bigl(x,u^\rho_t(t,x),\rD_xu^\rho(t,x),\rD_x^2u^\rho(t,x)\bigr),
    u^\rho(t,x) - \cM[u^\rho]^*(t,x)\Bigr\}
    + \frac{1}{\rho}\kappa(t,x) = 0,\; (t,x)\in\S_T.
\]
Similarly, if $v$ is a viscosity supersolution of \eqref{eq:HJBQVI}, the perturbation
$v^{\rho}$ is a viscosity supersolution of the perturbed equation
\[
  \min\Bigl\{F\bigl(x,v^\rho_t(t,x),\rD_xv^\rho(t,x),\rD_x^2v^\rho(t,x)\bigr),
    v^\rho(t,x) - \cM[v^\rho]_*(t,x)\Bigr\}
    - \frac{1}{\rho}\kappa(t,x) = 0,\; (t,x)\in\S_T.
\]
\end{lemma}

The role of the perturbation $u^\rho$ is to to control the growth of the subsolution
$u$ without sacrificing the subsolution property. Using this, we can establish
the following comparison principle.

\begin{theorem}\label{th:comparison-principle}
Let $u\in\USC$ be a viscosity subsolution and $v\in\LSC$ be a viscosity supersolution
of \eqref{eq:HJBQVI} for which there exists $K>0$ such that
\begin{equation}\label{eq:comparison_growth}
  0 \leq u(t,x),v(t,x) \leq K\bigl(1+|x|^{1-\alpha}\bigr),
  \quad (t,x)\in\overline{\S}_T.
\end{equation}
If, in addition, $u$ and $v$ satisfy the boundary conditions
\begin{equation}\label{eq:boundary}
  u(T,x)\leq v(T,x),\quad x\in\S,
  \quad\text{and}\quad
  u(t,0,p) = 0,\quad (t,p)\in[0,T]\times\Delta^{N},
\end{equation}
then $u\leq v$ on all of $\overline{\S}_T$.
\end{theorem}

\subsection{Construction of the Viscosity Solution}

We are now ready for the construction of the unique viscosity solution
of the HJBQVI. The ideas follow \cite{belak2019utility,belak2017general}
and especially \cite{belak2022optimal} with some minor additional
considerations required due to the presence of the classical control, i.e.\ the
trading strategy $\pi$. We begin by introducing a suitable notion of stochastic
supersolutions.

\begin{definition}\label{def:Stoch-Super}
The set of stochastic supersolutions of \eqref{eq:HJBQVI}, denoted by $\cV^{+}$,
is the set of functions $h:\overline{\S}_T\to\R$ satisfying
\begin{enumerate}
\item[\emph{($\cV^+_1$)}] $h\in\USC$;
\item[\emph{($\cV^+_2$)}] $h$ is lower bounded and there exists $K>0$ such that
\[
  h(t,x) \leq K\bigl(1 + |x|^{1-\alpha}\bigr),
  \quad (t,x)\in\overline{\S}_T;
\]
\item[\emph{($\cV^+_3$)}] $h$ satisfies the terminal inequality
\[
  h(T,x) \geq U(w),\quad x=(p,w)\in\S;
\]
\item[\emph{($\cV^+_4$)}] for each $(t,x)\in\overline{\S}_T$, each admissible strategy
$u=(\pi,\nu)\in\cA(x)$,
each $k\in\N$, each pair of $\cY^{\nu}$-stopping times $\theta\leq\rho$ taking values in
$[\tilde\tau_k,\tilde\tau_{k+1}]\cap[t,T]$, and each $\cY^\nu_\theta$-measurable
$\S$-valued random variable $\xi$ with $\E[|\xi|^2]<\infty$, it holds that
\[
  h(\theta,\xi) \geq \E\bigl[h(\rho,X^{u;\theta,\xi}_{\rho-})\big|\cY^\nu_\theta\bigr];
\]
\item[\emph{($\cV^+_5$)}] $h \geq \cM[h]$ on $\overline{\S}_T$.
\end{enumerate}
\end{definition}

It is easy to see that $\cV^+$ is non-empty. Indeed, as emphasized in
Remark~\ref{rem:stoch_supersolution_member}, the function $\psi$ with $\beta=\alpha$
and $A=0$ is a supersolution of the HJBQVI if $C>0$ is chosen sufficiently large.
This function obviously satisfies $(\cV^+_1)$, it satisfies the growth condition
$(\cV^+_2)$
\[
  \psi(t,x) \leq K\bigl(1 + |x|^{1-\alpha}\bigr),\quad (t,x)\in\overline{\S}_T,
\]
for $K:=e^{CT}/(1-\alpha)$ and the boundary conditions $(\cV^+_3)$
\[
  \psi(T,x) = U(w),\quad x=(w,p)\in\S,
  \quad\text{and}\quad
  \psi(t,0) = 0,\quad t\in[0,T].
\]
Finally, $(\cV^+_4)$ is an immediate consequence of an application of It\^{o}'s formula
and using that $\psi$ is a supersolution of the HJBQVI, whereas $(\cV^+_5)$ follows
directly from the supersolution property and the fact that $\psi(T,x)=U(w)$ is
strictly increasing in $w$.

Moreover, let us briefly show that any $h\in\cV^+$ dominates the value function.
For this, fix $(t,x)\in\overline{\S}_T$ and $u\in\cA(x)$. By $(\cV^+_4)$ and pathwise uniqueness,
we see that
\[
  \E\bigl[h\bigl(\tilde\tau_k,X^{u;t,x}_{\tilde\tau_k}\bigr)\bigr]
  \geq \E\bigl[h\bigl(T\wedge\tilde{\tau}_{k+1},X^{u;t,x}_{T\wedge\tilde{\tau}_{k+1}-}\bigr)\bigr],
  \quad\text{on }\{t\leq\tilde{\tau}_k\leq T\},
  \quad k\in\N_0,
\]
where $\tilde\tau_{0}:=0$. Moreover, from multiple applications of $(\cV^+_5)$ it follows that
\[
	\E\bigl[h\bigl(T\wedge\tilde{\tau}_{k+1},X^{u;t,x}_{T\wedge\tilde{\tau}_{k+1}-}\bigr)\bigr]
    \geq \E\bigl[h\bigl(T\wedge\tilde{\tau}_{k+1},X^{u;t,x}_{T\wedge\tilde{\tau}_{k+1}}\bigr)\bigr].
\]
Combining these two inequalities and finally using the terminal inequality $(\cV^+_3)$,
we arrive at
\[
  h(t,x) \geq \E\bigl[h\bigl(T,X^{u;t,x}_T\bigr)\bigr]
  \geq \E\bigl[U\bigl(W^{u;t,x}_T\bigr)\bigr],
\]
and maximizing the right-hand side over $u\in\cA(x)$ yields $h(t,x)\geq V(t,x)$ as claimed.

The pointwise infimum of all stochastic supersolutions is the function
\[
  V^+ :\overline{\S}_T \to [0,\infty),
  \quad (t,x)\mapsto V^+(t,x) := \inf_{h\in\cV^+} h(t,x).
\]
We show in this section that $V^+$ is the unique continuous viscosity solution
of the HJBQVI. From the discussion above, we already know that $V^+ \geq V\geq 0$,
where $V$ is the value function of the full information problem. Moreover, by
$(\cV^+_3)$ and by the properties of the particular stochastic supersolution $\psi\in\cV^+$,
it follows that
\[
  V^+(T,x) = U(w),\quad x=(w,p)\in\S,
  \quad\text{and}\quad
  V^+(t,0,p) = 0,\quad (t,p)\in[0,T]\times\Delta^N.
\]
Finally, by \cite[Proposition 4.1]{bayraktar2012stochastic}, the infimum in the
definition of $V^+$ can be restricted to a countable subset of $\cV^+$, from
which it is straightforward to see that $V^+\in\cV^+$, i.e.\ $V^+$ is actually
the pointwise minimum of the members of $\cV^+$. Using the classical stochastic
Perron machinery, one can show that $V^+$ is a viscosity subsolution of the HJBQVI.

\begin{theorem}\label{th:V+-is-viscosity-subsolution}
$V^{+}$ is a viscosity subsolution of \eqref{eq:HJBQVI} satisfying
\[
  0 \leq V^+(t,x) \leq K\bigl(1 + |x|^{1-\alpha}\bigr),
  \quad (t,x)\in\overline{\S}_T,
\]
for a constant $K>0$ and such that
\begin{equation}\label{eq:boundary_condition}
  V^+(T,x) = U(w),\quad x=(w,p)\in\S,
  \quad\text{and}\quad
  V^+(t,0,p) = 0,\quad (t,p)\in[0,T]\times\Delta^N.
\end{equation}
\end{theorem}

Conversely, it is not difficult to show that $V^+$ is also a viscosity supersolution
of the HJBQVI. This is in fact an immediate consequence of the properties of members
of $\cV^+$ and hence true for a much larger class of functions.

\begin{proposition}\label{th:V+-is-visco-super}
Every Borel-measurable function $h:\overline{\S}_T\to\R$ satisfying $(\cV^{+}_{2})$ to $(\cV^{+}_{5})$
in the definition of
stochastic supersolutions is a viscosity supersolution of the HJBQVI. In particular,
this is true for the choice of $h = V^+$, in which case it furthermore holds that
\[
  0 \leq (V^+)_*(t,x) \leq K\bigl(1 + |x|^{1-\alpha}\bigr),\quad (t,x)\in\overline{\S}_T,
\]
for some constant $K>0$ and
\[
  (V^+)_*(T,x) = U(w),\quad x=(w,p)\in\S.
\]
\end{proposition}

Combining the subsolution property in Theorem~\ref{th:V+-is-viscosity-subsolution}
and the supersolution property in Proposition~\ref{th:V+-is-visco-super}
with the comparison principle in Theorem~\ref{th:comparison-principle} yields the
final conclusion of this section, the characterization of $V^+$ as the unique
continuous viscosity solution of the HJBQVI.

\begin{corollary}\label{th:V+-is-viscosity-solution}
$V^+$ is the unique continuous viscosity solution of \eqref{eq:HJBQVI} in the
class of non-negative functions satisfying the growth condition~\eqref{eq:comparison_growth}
and the boundary condition~\eqref{eq:boundary_condition}.
\end{corollary}

\section{An Exit Time Control Problem and Optimal Strategies}\label{sec:exit-time-problem}

The HJBQVI suggests that the optimal expert opinion strategy is characterized
by the partition
\[
  \cont := \bigl\{(t,x)\in\overline{\S}_T : V^+(t,x) > \cM[V^+](t,x)\bigr\},
  \quad
  \interv := \bigl\{(t,x)\in\overline{\S}_T : V^+(t,x) = \cM[V^+](t,x)\bigr\}
\]
of the state space $\overline{\S}_T$ in the sense that in the ``continuation''
region $\cont$, it is strictly suboptimal to purchase an expert opinion whereas
it is optimal to purchase an expert opinion as soon as the time-augmented state
process $(t,X_t)$ enters the ``intervention'' region $\interv$. As wealth is finite,
purchases of expert opinions will eventually take the state process back into the
continuation region $\cont$. In particular, the optimally controlled state process
never spends a positive amount of time in $\interv$. To construct an optimal trading
strategy $\pi$, it should therefore suffice to restrict to the continuation
region $\cont$. This section paves the way for this argument by introducing a
suitable exit time control problem on $\cont$ and showing that its value function
coincides with the value function $V$ of the full information problem and the
pointwise minimum $V^+$ of the stochastic supersolutions.

\subsection{The Exit Time Control Problem and its Viscosity Characterization}

Throughout this section, we denote by
\[
  \cA^\circ(x) := \bigl\{\pi : (\pi,\circ)\in\cA(x)\bigr\}
\]
the set of admissible trading strategies corresponding to a fixed expert opinion
strategy $\circ$ which does not purchase any expert opinions on $[0,T]$. With
the expert opinion strategy fixed, we subsequently omit $\circ$ from our notation
of the state processes. Another consequence of $\circ$ being fixed is that the
innovations process $I=I^\circ$ and the filtration $\cY = \cY^\circ$ are fixed
in the present setting. Finally, note that $\cA^\circ(x)$ does not depend on $x$,
so we subsequently write $\cA^\circ$ for the set of admissible trading strategies of
the exit time control problem.

Next, we write $\partial\cont$ and $\overline{\cont}$
for the boundary and the closure of $\cont$, respectively. With this, let
\[
  \cS := \bigl\{ (t,x)\in\cont : t < T \bigr\}
  \quad\text{and}\quad
  \partial^*\cS := \bigl\{ (t,x)\in\partial\cont : t < T \bigr\}
    \cup \bigl\{(t,x)\in\overline{\cont} : t = T\bigr\}
\]
be the state space and the parabolic boundary of the state space of the exit
time control problem, respectively. In line with this notation, we furthermore
write $\overline{\cS} := \cS\cup\partial^*\cS = \overline{\cont}$ for the closure
of the state space. For any initial datum $(t,x)\in\overline{\cS}$ and any
admissible trading strategy $\pi\in\cA^\circ$, we write
\[
  \tau_{\cS}^\pi := \inf\bigl\{ s\in[t,T] : X^{\pi;t,x}_s \not\in\cS\bigr\}
\]
for the first exit time of $X^{\pi;t,x}$ from the state space. With this, the cost
functional and the value function of the exit time control problem are given by
\[
  \cJ^\cont : \cA^\circ \times \overline{\cS} \to [0,\infty),
  \quad (\pi,t,x)\mapsto \cJ^\cont(\pi;t,x) := \E\bigl[V^+\bigl(\tau_\cS^\pi,X^{\pi;t,x}_{\tau_\cS^\pi}\bigr)\bigr]
  \quad\text{subject to \eqref{eq:state-aux} with $\nu=\circ$.}
\]
and
\begin{equation}\label{eq:def-VF-exit}\tag{EP}
  V^\cont : \overline{\cS}\to[0,\infty),
  \quad
  (t,x) \mapsto V^\cont(t,x) := \sup_{\pi\in\cA^\circ}\cJ^\cont(\pi;t,x).
\end{equation}
Finally, for the sake of completeness, let us mention that the growth estimate
in Lemma~\ref{lem:quad-growth} is still valid for any state process $X^{\pi;t,x}$
since $(\pi,\circ)\in\cA(x)$ and $\overline{\cS}\subseteq\overline{\S}_T$.

The goal of this section is to show that $V^\cont = V^+$ on $\overline{\cS}$.
For this,
we once again employ a version of the stochastic Perron's method. More precisely,
we begin by introducing the set of stochastic subsolutions of the exit time control
problem and show that their pointwise maximum $V^-$ is a viscosity supersolution
of the Hamilton--Jacobi--Bellman (HJB) equation
\begin{equation}\label{eq:hjb_exit}\tag{HJB}
  F\bigl(x,V_t^-(t,x),\rD_xV^-(t,x),\rD_x^2V^-(t,x)\bigr) = 0,\quad (t,x)\in\cS,
\end{equation}
satisfying the boundary inequality
\[
  V^-(t,x) \geq V^+(t,x),\quad (t,x)\in\partial^*\cS.
\]
Since, conversely, the pointwise minimum of the set of stochastic subsolutions of
the full information problem $V^+$ is easily seen to be a viscosity subsolution
of this HJB, another comparison principle shows that $V^- = V^+$ on $\overline{\cS}$.
Finally, we argue that $V^-\leq V^\cont \leq V^+$ on $\overline{\cS}$,
from which the desired characterization of $V^\cont$ follows.

\begin{definition}\label{def:Stoch-Sub}
The set of stochastic subsolutions of \eqref{eq:hjb_exit}, denoted by $\cV^-$,
is the set of functions $h:\overline{\cS}\to\R$ satisfying
\begin{enumerate}
\item[\emph{($\cV^-_1$)}] $h\in\LSCExit$;
\item[\emph{($\cV^-_2$)}] $h$ is lower bounded and there exists $K>0$ such that
\[
  h(t,x) \leq K\bigl(1 + |x|^{1-\alpha}\bigr),\quad (t,x)\in\overline{\cS};
\]
\item[\emph{($\cV^-_3$)}] $h$ satisfies the boundary inequality
\[
  h(t,x) \leq V^+(t,x),\quad (t,x)\in\partial^*\cS;
\]
\item[\emph{($\cV^-_4$)}] for each stopping
time $\theta$ taking values in $[0,T]$ and each $\cY_\theta$-measurable
random variable $\xi$ with $\E[|\xi|^2]<\infty$, there exists $\pi\in\cA^\circ$
such that
\[
  h(\theta,\xi) \leq \E\bigl[h(\rho,X^{\pi;\theta,\xi}_{\rho})\big|\cY_\theta\bigr]
\]
for any stopping time $\rho$ with values in $[\theta,\tau^\pi_\cS]$.
\end{enumerate}
\end{definition}

Once again, it is straightforward to see that the set of stochastic subsolutions
is non-empty since, e.g., $0\in\cV^-$. It is furthermore straightforward to see
that each $h\in\cV^-$ is dominated by $V^\cont$. Indeed, for each $(t,x)\in\overline{\cS}$
we can apply $(\cV^-_4)$ with $(\theta,\xi):=(t,x)$ to find $\pi\in\cA^\circ$ such that,
with $\rho:=\tau^\pi_\cS$,
\[
  h(t,x) \leq \E\bigl[h(\tau^\pi_\cS,X^{\pi;t,x}_{\tau^\pi_\cS})\bigr]
  \leq \sup_{\pi\in\cA^\circ}\E\bigl[h(\tau^\pi_\cS,X^{\pi;t,x}_{\tau^\pi_\cS})\bigr]
  \leq \sup_{\pi\in\cA^\circ}\E\bigl[V^+(\tau^\pi_\cS,X^{\pi;t,x}_{\tau^\pi_\cS})\bigr]
  = V^\cont(t,x),
\]
where the last inequality is due to the boundary inequality $(\cV^-_3)$.
Next, we introduce the pointwise supremum
\[
  V^-:\overline{\cS}\to[0,\infty),
  \quad (t,x) \mapsto V^-(t,x) := \sup_{h\in\cV^-} h(t,x),
\]
and we proceed to show that $V^-$ is a viscosity subsolution of \eqref{eq:hjb_exit}.
As this is once again a classical stochastic Perron argument very similar to
\cite{bayraktar2013stochastic,rokhlin2014verification}, the proof is deferred to the
appendix.

\begin{theorem}\label{th:V--is-viscosity-supersolution}
$V^{-}$ is a viscosity supersolution of \eqref{eq:hjb_exit} satisfying
\[
  0 \leq V^-(t,x) \leq V^\cont(t,x) \leq V^+(t,x) \leq K\bigl(1 + |x|^{1-\alpha}\bigr),
  \quad (t,x)\in\overline{\cS},
\]
for a constant $K>0$ and such that
\[
  V^-(t,x) = V^+(t,x),\quad (t,x)\in\partial^*\cS.
\]
\end{theorem}

Using the viscosity supersolution property and given our results in
Section~\ref{sec:viscosity-characterization}, it is now straightforward to show that
$V^- = V^+ = V^\cont$ on $\overline{\cS}$. Indeed, we argue below that $V^+$ is a
viscosity subsolution of \eqref{eq:hjb_exit} and another comparison principle then
implies the result.

\begin{corollary}\label{cor:V+-restricted-subsolution}
The restriction of $V^+$ to $\overline{\cS}$ is a viscosity subsolution of
\eqref{eq:hjb_exit} on $\cS$.
\end{corollary}

\begin{proof}
This is an immediate consequence of the fact that $V^+$ is a viscosity subsolution
of \eqref{eq:HJBQVI} together with $V^+ > \cM[V^+]$ on $\cS$.
\end{proof}

Next, we observe that the comparison principle for $\eqref{eq:hjb_exit}$ follows from
a straightforward adaptation of the proof the comparison principle for $\eqref{eq:HJBQVI}$.

\begin{theorem}\label{th:comparison-principle-exit-time}
Let $u\in\USCExit$ be a viscosity subsolution and $v\in\LSCExit$ be a viscosity
supersolution of \eqref{eq:hjb_exit} for which there exists a constant $K>0$ such that
\[
  0 \leq u(t,x),v(t,x) \leq K\bigl(1 + |x|^{1-\alpha}\bigr),\quad (t,x)\in\overline{\cS}.
\]
If, in addition, $u$ and $v$ satisfy the boundary condition
\[
  u \leq v\quad\text{on }\partial^*\cS,
\]
then $u\leq v$ on all of $\overline{\cS}$.
\end{theorem}

\begin{proof}
The proof works analogously to the proof of Theorem~\ref{th:comparison-principle}.
The only differences are that any occurrence of $\S_T$ has to be replaced by $\cS$,
step 2, handling of the non-locality of \eqref{eq:HJBQVI}, can be skipped, and
instead of arguing that $(t^*,x^*)$ cannot satisfy $t^*=T$, in the present situation
we have to show that $(t^*,x^*)\not\in\partial^*\cS$. This, however, is an
immediate consequence of the boundary condition $u\leq v$ on $\partial^*\cS$.
\end{proof}

Combining the supersolution property of $V^-$ in
Theorem~\ref{th:V--is-viscosity-supersolution} with the subsolution property of $V^+$
in Corollary~\ref{cor:V+-restricted-subsolution} together with the boundary identity
$\V^- = V^+$ on $\partial^*\cS$ and the general set of inequalitites
$V^-\leq V^\cont \leq V^+$ on $\overline{\cS}$, the comparison principle in
Theorem~\ref{th:comparison-principle-exit-time} implies that these three functions
actually coincide and constitute the unique continuous viscosity solution of
\eqref{eq:hjb_exit}.

\begin{corollary}
It holds that $V^- = V^\cont = V^+$ on $\overline{\cS}$ and $V^\cont$ is the unique
continuous viscosity solution of \eqref{eq:hjb_exit} in the class of non-negative
functions satisfying the growth condition~\eqref{eq:comparison_growth} and which are
equal to $V^+$ on $\partial^*\cS$.
\end{corollary}

\subsection{On the Construction of Optimal Strategies}

We conclude this article by showing how to construct optimal strategies. The
central observation is that the connection of the full information optimal
investment problem to the exit time control problem considered in the
previous subsection allows us to separate the construction of the optimal
trading strategy as the solution of the exit time control problem from
the construction of the optimal expert opinion strategy using a variant of
the impulse control verification arguments developed in the last decade
in \cite{belak2019utility,belak2017general,belak2022optimal,christensen2014solution}.

There are, however, a few subtle technicalities to take into account in the
construction of an optimal strategy. For the impulse control verification
procedure, a central assumption is the strong Markov property of the candidate
optimal state process.
Since trading strategies are only assumed to be progressively measurable,
there is no guarantee that the state processes in our problem
are Markov processes.
In \cite{haussmann1990existence}, existence of optimal feedback controls
for a class of stochastic control problems including our exit time control
problem (up to a minor upper boundedness assumption, which is not impedient
for the validity of their results)
has been studied. However,
unfortunately for our purposes, the result is in general only valid for
a weak formulation of the control problem, that is, on a particular
probability space with a particular filtration and Brownian motion.
Since we are using the strong formulation here, we cannot apply their
results directly.

There are several ways to work around these issues. One possibility is to
formulated the original problem in a weak formulation to begin with,
another is to allow for additional independent randomness to lift the weak
optimal control to a strong one. Both of these techniques are well-understood,
but come with a technical burden which in our opinion would distract from the main
contributions of this article. We therefore \emph{assume} below
that the optimal feedback control constructed in \cite{haussmann1990existence}
is sufficiently regular so that the stochastic differential equation
for the optimally controlled state process admits a strong solution.

Towards a precise formulation of our assumption, let us begin by
recalling one of the central results in \cite{haussmann1990existence},
translated to the notation of our article. According to \cite{haussmann1990existence},
there exists a measurable function
\[
  \hat{\pi} : \cS \to \Pi,\qquad (t,x)\mapsto \hat{\pi}(t,x)
\]
such that, for any initial state $(t,x)\in\cS$, the stochastic
differential equation
\begin{equation}\label{eq:opt-state-process}
  \rd \hat{X}_s = f\bigl(\hat{X}_s,\hat{\pi}(s,\hat{X}_s)\bigr)\rd s
    + \Sigma\bigl(\hat{X}_s,\hat{\pi}(s,\hat{X}_s)\bigr)\rd B_s,
  \quad \hat{X}_t = x,\quad s\in[t,T],
\end{equation}
admits a \emph{weak solution} such that, on the probability space on
which this solution exists, the control $\pi := \hat{\pi}(\argdot,\hat{X})$
is optimal, $\hat{X}$ is the associated optimal state process, and
$\hat{X}$ is a strong Markov process.
We can extend $\hat{\pi}$ to a measurable function defined on all of
$[0,T]\times\R\times\R^{N}$ by setting $\hat\pi(t,x) := \pi_0$ for some fixed
$\pi_0\in\Pi$. We expect that the feedback
function $\hat{\pi}$ also gives rise to an optimal trading strategy for the
full information problem. However, since we do not have any regularity of
$\hat{\pi}$ beyond measurability, the existence of a (strong) solution on our
fixed probability space with the innovations Brownian motion as driving noise
is, in general, not clear.

\begin{assumption}
The stochastic differential equation~\eqref{eq:opt-state-process} admits a
pathwise unique strong solution.
\end{assumption}

In particular, we can solve~\eqref{eq:opt-state-process} on our fixed probability
space for any choice of driving Brownian motion. Given a $[0,T]$-valued $\F$-stopping time
$\tau$, an $\R\times\R^N$-valued $\cF_\tau$-measurable random variable
$\xi$, and a Brownian motion $B$, we denote by
$\hat{X}^{\tau,\xi,B} = (\hat{X}^{\tau,\xi,B}_s)_{s\in[\tau,T]}$ the strong
solution started in $(\tau,\xi)$ with driving Brownian motion $B$.

Next, let us turn to the candidate for the optimal expert opinion strategy.
For this, recall the sets
\[
  \cont := \bigl\{(t,x)\in\overline{\S}_T : V^+(t,x) > \cM[V^+](t,x)\bigr\},
  \quad
  \interv := \bigl\{(t,x)\in\overline{\S}_T : V^+(t,x) = \cM[V^+](t,x)\bigr\}.
\]
The intuition of these sets is that in $\cont$, purchases of expert opinions are
strictly suboptimal whereas purchases of expert opinions are optimal in $\interv$
provided that the quality level is chosen as a maximizer of $\cM[V^+]$.
Since $V^+$ is continuous and $\cD(t,x)$ is compact, it follows from a standard
measurable selection argument as for example \cite[Theorem 5.3.1]{srivastava2008course}
that there exists a measurable
function
\[
  \hat{q}: \overline{\S}_T\setminus\overline{\S}_T^\emptyset \to [0,1],
  \quad (t,x) \mapsto \hat{q}(t,x)
\]
such that, for all $(t,x)\in\overline{\S}_T\setminus\overline{\S}_T^\emptyset$,
\[
  \hat{q}(t,x)\in\cD(t,x)
  \;\text{ and }\;
  \cM[V^+](t,x) = \sum_{n=1}^N p^n\E\Bigl[
    V^+\Bigl(t,\Gamma\bigl(\hat{q}(t,x)\mu^n + (1-\hat{q}(t,x))\cN,t,x,\hat{q}(t,x)\bigr)\Bigr)\Bigr].
\]
With the feedback maps $\hat{\pi}$ and $\hat{q}$ at hand, we can now
introduce the candidate optimal strategy $u^*=(\pi^*,\nu^*)$ with
$\nu^*=(\tau^*,q^*)$ as follows.
For an initial state $(t,x)\in\overline{\S}_T$, we set
$(\tau_0^*,\xi_0^*,I_0^*) := (t,x,I^\circ)$ and iteratively for all $k\in\N$
\begin{align*}
  \hat{X}^k &:= \hat{X}^{\tau_{k-1}^*,\xi_{k-1}^*,I_{k-1}^*},
  &
  \tau_k^* &:= \inf\bigl\{s\in[\tau_{k-1}^*,T] :
    \bigl(s,\hat{X}^{k}_s\bigr) \in \interv\bigr\},\\
  q_k^* &:= \hat{q}\bigl(\tau_k^*,\hat{X}^k_{\tau_k^*}\bigr)
    \ind_{\{\tau_k^*\leq T\}},
  &
  \xi_k^* &:= \Gamma\bigl(q_k^*\mu^\intercal Y_{\tau_k^*}
    + (1-q_k^*)\cN_k,\tau_k^*,\hat{X}^k_{\tau_k^*},q_k^*\bigr)
    \ind_{\{\tau_k^*\leq T\}},\\
  \nu_k^* &:= (\tau_j^*,q_j^*)_{0\leq j\leq k},
  &
  I^*_k &:= I^{\nu_k^*},
\end{align*}
and, with this,
\[
  \pi^* := \sum_{k=1}^\infty
    \hat{\pi}(\argdot,\hat{X}^k)\ind_{[\tau_{k-1}^*,\tau_k^*)\cap[0,T]}.
\]
To understand what is going on in the construction of $u^*$, let us first
write $X^* = (W^*,p^*)$ as short-hand notation for $X^{u^*;t,x}$. As long
as $X^*\in\cont$, no expert opinions are being purchased and the agent
follows the optimal exit time control $\hat{\pi}$. In particular, we have
$X^* = \hat{X}^{k}$ on $[\tau_k,\tau_{k+1})\cap[0,T]$. Whenever $X^*$
hits $\interv$, an expert opinion of quality $\hat{q}$ is purchased. Observe
that such a purchase causes a jump in $X^*$, but it is in general not
guaranteed that this jump takes $X^*$ back to $\cont$. However, since
each purchase results in a loss of wealth of at least $K_{\min}$ and
since $\overline{\S}_T^\emptyset\subset\cont$, it follows that there
can be at most finitely many purchases at any point in time before the
state process jumps back to $\cont$.

Next, it is clear that $\pi^*$ is a $\Pi$-valued and $\F$-progressively
measurable process and that $\nu^*\in\cA_{pre}^{exp}$. In fact, writing
$\cY^* := \cY^{\nu^*}$ and $\cY^{*,k} := \cY^{\nu^*,k}$, it is clear that
$\pi^*$ is even $\cY^*$-progressive, $\tau_k^*$ is a $\cY^{*,k-1}$-stopping
time, and $q_k^*$ is $\cY^{*,k-1}_{\tau_k^*}$-measurable for all $k\in\N$.
Finally, since $X^* = \hat{X}^{k}$ on $[\tau_k,\tau_{k+1})\cap[0,T]$ and
by the properties of $\Gamma$, we see that $W^*\geq 0$. We have therefore
argued that $u^*\in\cA(w)$, that is, $\pi^*$ is an admissible trading
strategy and $\nu^*$ is an admissible expert opinion strategy.
With this, we can prove the main result of this article showing that $u^*$
is indeed optimal and $V = V^+$, hence characterizing $V$ as the unique
continuous viscosity solution of \eqref{eq:HJBQVI}.

\begin{theorem}
For any $(t,x)\in\overline{\S}_T$, the associated strategy $u^*$ is
optimal for the full information optimal investment problem and
$V(t,x) = V^+(t,x)$.
\end{theorem}

\begin{proof}
Step 1. We show that $V^+(\argdot,\hat{X}^k)$ is a martingale on
$[\tau^*_{k-1},\tau^*_{k}]\cap[t,T]$ for each $k\in\N$. For this, we first
observe that
\[
  V^+(\argdot,\hat{X}^k) = V^\cont(\argdot,\hat{X}^k)
  \quad\text{on }[\tau^*_{k-1},\tau^*_{k}]\cap[t,T]
\]
since $\hat{X}^k$ is $\overline{\cS}$-valued on that time interval.
Moreover, $\hat{X}^k$ is an optimally controlled state process
for the exit time control problem with initial state
$(\tau_{k-1}^*,\hat{X}^k_{\tau_{k-1}^*})$. The martingale property
of $V^*$ therefore follows from the martingale optimality principle.
Indeed, let $\rho$ be an arbitrary stopping time taking values
in $[\tau^*_{k-1},\tau^*_{k}]\cap[t,T]$ whenever this interval
is non-empty and let $\rho$ be equal to $\tau^*_{k-1}$ otherwise.
By optimality of $\hat{\pi}$ and pathwise uniqueness of $\hat{\pi}$, we see that
\begin{align*}
  \E\bigl[V^\cont(\rho,\hat{X}^k_\rho)\bigr]
  &= \E\Bigl[\E\bigl[V^+\bigl(\tau^{\hat{\pi}}_{\cS},
    \hat{X}^{t,x,I^\circ}_{\tau^{\hat{\pi}}_{\cS}}\bigr)\bigr]_{(t,x)=(\rho,\hat{X}^k_\rho)}\Bigr]\\
  &= \E\Bigl[\E\bigl[V^+\bigl(\tau^{\hat{\pi}}_{\cS},
    \hat{X}^{t,x,I^*_{k-1}}_{\tau^{\hat{\pi}}_{\cS}}\bigr)\bigr]_{(t,x)=(\rho,\hat{X}^k_\rho)}\Bigr]\\
  &= \E\Bigl[V^+\bigl(\tau^{\hat{\pi}}_{\cS},
    \hat{X}^{\rho,\hat{X}^k_\rho,I^*_{k-1}}_{\tau^{\hat{\pi}}_{\cS}}\bigr)\Bigr]\\
  &= \E\Bigl[V^+\bigl(\tau^{\hat{\pi}}_{\cS},
    \hat{X}^{\tau_{k-1}^*,\hat{X}^k_{\tau^*_{k-1}},I^*_{k-1}}_{\tau^{\hat{\pi}}_{\cS}}\bigr)\Bigr]
  = \E\Bigl[V^+\bigl(\tau^{\hat{\pi}}_{\cS},
    \hat{X}^{k}_{\tau^{\hat{\pi}}_{\cS}}\bigr)\Bigr].
\end{align*}
In particular, the same identity is true for $\rho:=\tau_{k-1}^*$, hence
\[
  \E\bigl[V^\cont(\rho,\hat{X}^k_\rho)\bigr]
  = \E\bigl[V^\cont(\tau_{k-1}^*,\hat{X}^k_{\tau_{k-1}^*})\bigr],
\]
and we conclude by optional sampling.

Step 2. For each $k\in\N$, we show that
\[
  \E\bigl[V^+(\tau_k^*,\hat{X}^k_{\tau_k^*})\ind_{\{\tau_k^*<T\}}\bigr]
  = \E\bigl[V^+(\tau_k^*,\hat{X}^{k+1}_{\tau_k^*})\ind_{\{\tau_k^*<T\}}\bigr].
\]
For this, we first observe that $\hat{X}^k_{\tau_k^*}\in\interv$ on
$\{\tau_k^*<T\}$. Thus, by choice of $q_k^*$,
\begin{align*}
  \E\bigl[V^+(\tau_k^*,\hat{X}^k_{\tau_k^*})\ind_{\{\tau_k^*<T\}}\bigr]
  &= \E\Bigl[\cM[V^+](\tau_k^*,\hat{X}^k_{\tau_k^*})\ind_{\{\tau_k^*<T\}}\Bigr]\\
  &= \E\Bigl[V^+\Bigl(\tau_k^*,\Gamma\bigl(Z^{\nu^*}_k,\tau_k^*,\hat{X}^k_{\tau_k^*},q_k^*\bigr)\Bigr)
    \ind_{\{\tau_k^*<T\}}\Bigr]\\
  &= \E\bigl[V^+(\tau_k^*,\xi_k^*)\ind_{\{\tau_k^*<T\}}\bigr]
  = \E\bigl[V^+(\tau_k^*,\hat{X}^{k+1}_{\tau_k^*})\ind_{\{\tau_k^*<T\}}\bigr].
\end{align*}

Step 3. Conclusion. Using the previous two steps, we see that
\begin{align*}
  V^+(t,x)
  &= \E\Bigl[V^+(\tau_1^*,\hat{X}^1_{\tau_1^*})\ind_{\{\tau^*_1 < T\}}
    + V^+(T,\hat{X}^1_T)\ind_{\{\tau^*_1 \geq T\}}\Bigr]\\
  &= \E\Bigl[V^+(\tau_1^*,\hat{X}^{2}_{\tau_1^*})\ind_{\{\tau^*_1 < T\}}
    + V^+(T,\hat{X}^1_T)\ind_{\{\tau^*_1 \geq T\}}\Bigr]\\
  &= \E\Bigl[V^+(\tau_2^*,\hat{X}^{2}_{\tau_2^*})\ind_{\{\tau^*_1 < T\}}\ind_{\{\tau^*_2<T\}}
    + V^+(T,\hat{X}^{2}_{T})\ind_{\{\tau^*_1 < T\}}\ind_{\{\tau_2^*\geq T\}}
    + V^+(T,\hat{X}^1_T)\ind_{\{\tau^*_1 \geq T\}}\Bigr]\\
  &= \E\Bigl[V^+(\tau_2^*,\hat{X}^{2}_{\tau_2^*})\ind_{\{\tau^*_2<T\}}
    + V^+(T,\hat{X}^{2}_{T})\ind_{\{\tau^*_1 < T\}}\ind_{\{\tau_2^*\geq T\}}
    + V^+(T,\hat{X}^1_T)\ind_{\{\tau^*_1 \geq T\}}\Bigr].
\end{align*}
Now since $V^+(T,w,p) = U(w)$ and $U > \cM[U]$, we conclude that
$\{T\}\times\R_+\times\R^N\subset\cont$ and hence $\hat{X}^1_T = X^*_T$ on
$\{\tau^*_1 \geq T\}$ and $\hat{X}^{2}_{T} = X^*_T$ on $\{\tau^*_1 < T\}\cap\{\tau_2^*\geq T\}$.
Thus, in particular,
\begin{align*}
  V^+(t,x)
  &= \E\Bigl[V^+(\tau_2^*,\hat{X}^{2}_{\tau_2^*})\ind_{\{\tau^*_2<T\}}
    + U(W^*_T)\ind_{\{\tau^*_1 < T\}}\ind_{\{\tau_2^*\geq T\}}
    + U(W^*_T)\ind_{\{\tau^*_1 \geq T\}}\Bigr]\\
  &= \E\Bigl[V^+(\tau_2^*,\hat{X}^{2}_{\tau_2^*})\ind_{\{\tau^*_2<T\}}
    + U(W^*_T)\ind_{\{\tau_2^*\geq T\}}\Bigr].
\end{align*}
Iterating this argument yields
\[
  V^+(t,x) = \E\Bigl[V^+(\tau_k^*,\hat{X}^{k}_{\tau_k^*})\ind_{\{\tau^*_k<T\}}
    + U(W^*_T)\ind_{\{\tau_k^*\geq T\}}\Bigr],
  \quad k\in\N.
\]
Sending $k\to\infty$, using dominated convergence and the fact that $\lim_{k\to\infty}\tau_k^* > T$,
it follows that
\begin{multline*}
  V^+(t,x)
  = \lim_{k\to\infty} \E\Bigl[V^+(\tau_k^*,\hat{X}^{k}_{\tau_k^*})\ind_{\{\tau^*_k<T\}}
    + U(W^*_T)\ind_{\{\tau_k^*\geq T\}}\Bigr]\\
  = \E\bigl[U(W^*_T)\bigr] = \cJ(u^*;t,x) \leq V(t,x) \leq V^+(t,x),
\end{multline*}
which concludes the proof.
\end{proof}

\appendix

\section{Proofs related to the Stochastic Perron's Method}
\label{app:perron}

\paragraph{Proof of Lemma~\ref{lem:regularity-intervention}.}

Step 1. Suppose that $v\in\USC$. Let $(t,x)\in\overline{\S}_T$ and
$\{(t_{k},x_{k})\}_{k\in\N}\subset\overline{\S}_T$ with $(t_{k},x_{k})\to(t,x)$.
As by definition $\cM[v]=-1$ is constant on $\overline{\S}_T^\emptyset$ which is
open relative to $\overline{\S}_T$, we may assume that
$(t,x)\in\overline{\S}_T\setminus\overline{\S}_T^\emptyset$.
After dropping to a subsequence if necessary, we can distinguish between the two
cases
\[
  \{(t_{k},x_{k})\}_{k\in\N}\subset\overline{\S}_T^\emptyset
  \quad\text{or}\quad
  \{(t_{k},x_{k})\}_{k\in\N}\subset\overline{\S}_T\setminus\overline{\S}_T^\emptyset.
\]
In the former case we make use of the fact that $\cD(t_k,x_k) = \emptyset$ for
all $k\in\N$ but $\cD(t,x)\neq\emptyset$ and hence
\[
  \limsup_{k\to\infty} \cM[v](t_k,x_k) = -1 < 0 \leq \cM[v](t,x)
\]
since $v\geq 0$. In the other case we have $\cD(t_{k},x_{k})\neq\emptyset$ for all
$k\in\N$ and since $v$ is upper semi-continuous and $\cD(t_{k},x_{k})$ compact,
there exists $q_k\in\cD(t_{k},x_{k})$ such that
\[
  \cM[v](t_{k},x_{k}) = \sum_{n=1}^N p^n_k
    \E\Bigl[v\Bigl(t_k,\Gamma\bigl(q_k\mu^n+(1-q_k)\cN,t_k,x_k,q_k\bigr)\Bigr)\Bigr].
\]
Since $\cD(t_{k},x_{k}) = [0,\chi(t_k,x_k)]$ and $\chi$ is continuous, it follows
that, upon dropping to a subsequence if necessary, $q_k\to q$ for some
$q\in[0,\chi(t,x)]$.
Next, we observe that with $\hat{w}:=\sup_{k\in\N} w_k$ we have
\[
  \Gamma\bigl(q_k\mu^n+(1-q_k)\cN,t_k,x_k,q_k\bigr)
  \in [0,\hat{w}]\times\Delta^N,
  \quad k\in\N,
\]
showing that this sequence of random variables is bounded.
But then upper semi-continuity of $v$, continuity of $\Gamma$, and Fatou's lemma imply
\begin{align*}
  \limsup_{k\to\infty} \cM[v](t_k,x_k)
  &= \limsup_{k\to\infty} \sum_{n=1}^N p^n_k
    \E\Bigl[v\Bigl(t_k,\Gamma\bigl(q_k\mu^n+(1-q_k)\cN,t_k,x_k,q_k\bigr)\Bigr)\Bigr]\\
  &\leq \sum_{n=1}^N p^n
    \E\Bigl[v\Bigl(t,\Gamma\bigl(q\mu^n+(1-q)\cN,t,x,q\bigr)\Bigr)\Bigr]
  \leq \cM[v](t,x),
\end{align*}
which concludes the case of $v\in\USC$.

Step 2. Suppose that $v\in\LSC$.
Since $\cM[v]=-1$ is constant on $\overline{\S}_T^\emptyset$,
we only have to consider the case in which $\cM[v]$ is restricted to
$\overline{\S}_T\setminus\overline{\S}_T^\emptyset$.
Let us therefore fix
$(t,x)\in\overline{\S}_T\setminus\overline{\S}_T^\emptyset$ as well as a sequence
$\{(t_k,x_k)\}_{k\in\N}\subset\overline{\S}_T\setminus\overline{\S}_T^\emptyset$
such that $(t_k,x_k)\to(t,x)$ as $k\to\infty$.
For any $q\in\cD(t,x)=[0,\chi(t,x)]$,
continuity of $\chi$ implies that there exists for each $k\in\N$ some
$q_k\in\cD(t_k,x_k)=[0,\chi(t_k,x_k)]$ such that $q_k\to q$.
But then lower semi-continuity of $v$, continuity of
$\Gamma$, and Fatou's lemma show that
\begin{align*}
  \liminf_{k\to\infty} \cM[v](t_k,x_k)
  &\geq \liminf_{k\to\infty} \sum_{n=1}^N p^n_k
  \E\Bigl[v\Bigl(t_k,\Gamma\bigl(q_k\mu^n+(1-q_k)\cN,t_k,x_k,q_k\bigr)\Bigr)\Bigr]\\
  &\geq \sum_{n=1}^N p^n
  \E\Bigl[v\Bigl(t,\Gamma\bigl(q\mu^n+(1-q)\cN,t,x,q\bigr)\Bigr)\Bigr].
\end{align*}
As this holds for any $q\in\cD(t,x)$, we conclude that $\cM[v]$ restricted to
$\overline{\S}_T\setminus\overline{\S}_T^\emptyset$ is lower semi-continuous
as claimed.\fakeQED

\paragraph{Proof of Lemma~\ref{lem:strict-classical-super}.}

As $1-\beta \geq 1-\alpha$, we immediately see that $\psi(T,w,p) \geq U(w)$ for all
$(w,p)\in\S$ provided that $A>0$ is sufficiently large.
For $(t,x)=(t,w,p)\in\S_T$, the partial derivatives of $\psi$ are given by
\[
	\psi_{t}(t,x) = -C\Bigl(A + \frac{1}{1-\beta}w^{1-\beta}\Bigr)e^{C(T-t)},
	\quad
	\psi_{w}(t,x) = w^{-\beta}e^{C(T-t)},
	\quad
	\psi_{ww}(t,x) = -\beta w^{-\beta-1}e^{C(T-t)},
\]
from which we compute, for any $\pi\in\Pi$,
\begin{align*}
  -\psi_t(t,x) - \cH\bigl(x,\pi,\rD_x\psi(t,x),\rD^2_x\psi(t,x)\bigr)
  &= -\psi_t(t,x) - \pi w \mu^\intercal p \psi_w(t,x)
    - \frac{1}{2}\pi^2w^2\sigma^2\psi_{ww}(t,x)\\
  &= \Bigl[AC + \Bigl(\frac{C}{1-\beta}
    - \pi \mu^\intercal p
	+ \frac{1}{2}\pi^2\sigma^2\beta\Bigr)w^{1-\beta}\Bigr]e^{C(T-t)}\\
  &\geq \Bigl[AC + \Bigl(\frac{C}{1-\beta} - \pi_{\max}\mu_{\max}\Bigr)
    w^{1-\beta}\Bigr]e^{C(T-t)},
\end{align*}
where $\pi_{\max} := \max\{|\underline{\pi}|,|\overline{\pi}|\}$ and
$\mu_{\max} := \max_{n=1,\dots,N}|\mu^n|$. Now if $C > (1-\beta)\pi_{\max}\mu_{\max}$,
we conclude that
\[
  F\bigl(x,\psi_t(t,x),\rD_x\psi(t,x),\rD^2_x\psi(t,x)\bigr)
  = -\psi_t(t,x) - \max_{\pi\in\Pi}\cH\bigl(x,\pi,\rD_x\psi(t,x),\rD^2_x\psi(t,x)\bigr)
  \geq AC > 0.
\]
Regarding the non-local part of the HJBQVI, we first observe that
\[
  \psi(t,x) - \cM[\psi](t,x) \geq 1>0,\quad (t,x)\in\overline{\S}_T^\emptyset.
\]
On the other hand, if $(t,x)=(t,w,p)\in\overline{\S}_T$ with $\cD(t,x)\neq\emptyset$,
i.e.\ such that $w - K(t,0)\geq 0$, then
\begin{align*}
  \psi(t,x) - \cM[\psi](t,x)
  &= \frac{1}{1-\beta}\Bigl(w^{1-\beta}
    - \sup_{q\in\cD(t,x)}\bigl(w-K(t,q)\bigr)^{1-\beta}\Bigr)e^{C(T-t)}\\
  &\geq w^{1-\beta} - \bigl(w-K_{\min}\bigr)^{1-\beta}
  > 0.
\end{align*}
But this implies the existence of a function $\kappa$ as claimed.\fakeQED

\paragraph{Proof of Lemma~\ref{lem:perturbation}.}

We only establish the claim for $u^\rho$, the case of $v^\rho$ follows in the same
way with some minor adjustments. Note that $\cM$ preserves upper semi-continuity
by Lemma~\ref{lem:regularity-intervention}, so we can drop the semi-continuous
envelope around $\cM[u^\rho]$.
Fix $(t,x)\in\S_T$ and let $\varphi^{\rho}\in\cC^{1,2}(\S_T)$ be a test function
for the subsolution property of $u^{\rho}$ at $(t,x)$. This means that
$u^{\rho}\leq\varphi^{\rho}$ and $u^{\rho}(t,x)=\varphi^{\rho}(t,x)$.
Since $u^{\rho} = \frac{\rho+1}{\rho}u - \frac{1}{\rho}\psi$, we see that
\[
  u \leq \varphi := \frac{\rho}{\rho+1}\varphi^{\rho} + \frac{1}{\rho+1}\psi
  \quad\text{and}\quad
  u(t,x) = \varphi(t,x),
\]
from which we obtain that $\varphi$ is a test function for the subsolution property
of $u$ and hence
\begin{equation}\label{eq:perturbation_proof_1}
  \min\Bigl\{F\bigl(x,\varphi_t(t,x),\rD_x\varphi(t,x),\rD_x^2\varphi(t,x)\bigr),
    u(t,x) - \cM[u](t,x)\Bigr\}
  \leq 0.
\end{equation}
Using that $\varphi^{\rho} = \frac{\rho+1}{\rho}\varphi - \frac{1}{\rho}\psi$,
it follows that
\begin{align}
  &\mathrel{\phantom{=}}F\bigl(x,\varphi^\rho_t(t,x),\rD_x\varphi^\rho(t,x),\rD_x^2\varphi^\rho(t,x)\bigr)\notag\\
  &\hspace{2cm}\leq \frac{\rho+1}{\rho}F\bigl(x,\varphi_t(t,x),\rD_x\varphi(t,x),\rD_x^2\varphi(t,x)\bigr)
    - \frac{1}{\rho}F\bigl(x,\psi_t(t,x),\rD_x\psi(t,x),\rD_x^2\psi(t,x)\bigr)\notag\\
  &\hspace{2cm}\leq \frac{\rho+1}{\rho}F\bigl(x,\varphi_t(t,x),\rD_x\varphi(t,x),\rD_x^2\varphi(t,x)\bigr)
    - \frac{1}{\rho}\kappa(t,x).\label{eq:perturbation_proof_2}
\end{align}
On the other hand, we have
\begin{align}
  u^\rho(t,x) - \cM[u^\rho](t,x)
  &= \frac{\rho+1}{\rho}u(t,x) - \frac{1}{\rho}\psi(t,x)
    - \cM\Bigl[\frac{\rho+1}{\rho}u - \frac{1}{\rho}\psi\Bigr](t,x)\notag\\
  &\leq \frac{\rho+1}{\rho}\Bigl(u(t,x) - \cM[u](t,x)\Bigr)
    - \frac{1}{\rho}\Bigl(\psi(t,x) - \cM[\psi](t,x)\Bigr)\notag\\
  &\leq \frac{\rho+1}{\rho}\Bigl(u(t,x) - \cM[u](t,x)\Bigr)
    - \frac{1}{\rho}\kappa(t,x).\label{eq:perturbation_proof_3}
\end{align}
Combining \eqref{eq:perturbation_proof_2} and \eqref{eq:perturbation_proof_3}
with \eqref{eq:perturbation_proof_1} hence yields
\begin{multline*}
  \min\Bigl\{F\bigl(x,\varphi^\rho_t(t,x),\rD_x\varphi^\rho(t,x),\rD_x^2\varphi^\rho(t,x)\bigr),
    u^\rho(t,x) - \cM[u^\rho](t,x)\Bigr\} + \frac{1}{\rho}\kappa(t,x)\\
  \leq \frac{\rho+1}{\rho}\min\Bigl\{F\bigl(x,\varphi_t(t,x),\rD_x\varphi(t,x),\rD_x^2\varphi(t,x)\bigr),
  u(t,x) - \cM[u](t,x)\Bigr\}
\leq 0,
\end{multline*}
thus proving the claim.\fakeQED

\paragraph{Proof of Theorem~\ref{th:comparison-principle}.}

Step 1. Doubling of variables and Ishii's lemma.
Let $\psi$ and $\kappa$ be as in Lemma~\ref{lem:strict-classical-super}
and, for $\rho>1$, let $u^\rho$ and $v^\rho$ be the perturbations as in
Lemma~\ref{lem:perturbation}. Note that since $\beta < \alpha$ in the definition
of the strict supersolution, $\psi$ grows faster than $u$ and $v$ as $w\to\infty$.
In fact, by choosing the constant $A$ in the definition of $\psi$ large enough,
it follows from \eqref{eq:comparison_growth} that we may assume $u,v\leq\psi$.
We proceed to show that $u^\rho\leq v^\rho$ and hence the result follows as $\rho\to\infty$.
Towards a contradiction, let us assume that there exists $(t^*,x^*)\in\overline{\S}_T$
such that
\begin{align}\label{eq:proof-CP-contr-asspt}
  u^{\rho}(t^*,x^*) - v^{\rho}(t^*,x^*) > 0.
\end{align}
Next, for each $k\in\N_0$, we define $\varphi_k:\overline{\S}_T\times\overline{\S}_T\to\R$ as
\[
  \varphi_k(t,x,\hat{t},\hat{x})
  := u^{\rho}(t,x) - v^{\rho}(\hat{t},\hat{x})
    - \frac{k}{2}\bigl(|t-\hat{t}|^{2}+|x-\hat{x}|^{2}\bigr)
\]
and constants
\[
  \Theta_k := \sup_{(t,x),(\hat{t},\hat{x})\in\overline{\S}_T}
    \varphi_{k}(t,x,\hat{t},\hat{x})
  \quad\text{and}\quad
  \Theta := \sup_{(t,x)\in\overline{\S}_T} \varphi_{0}(t,x,t,x).
\]
We observe that
\begin{equation}\label{eq:comparison_proof_1}
  0 < u^{\rho}(t^{*},x^{*}) - v^{\rho}(t^{*},x^{*})
  \leq \Theta \leq \Theta_{k+1} \leq \Theta_{k}\leq \Theta_{0},
  \quad k\in\N,
\end{equation}
and, since $v^\rho\geq 0$ and by~\eqref{eq:comparison_growth},
\begin{align}
  \Theta_{0}
  &\leq \sup_{(t,x)\in\overline{\S}_T}\Bigl\{\frac{\rho+1}{\rho}u(t,x)
    - \frac{1}{\rho}\psi(t,x)\Bigr\}\notag\\
  &\leq \sup_{x=(w,p)\in\S}\Bigl\{\frac{(\rho+1)K}{\rho}\bigl(1+|x|^{1-\alpha}\bigr)
    - \frac{1}{\rho(1-\beta)}|w|^{1-\beta}\Bigr\}
  <\infty,\label{eq:comparison_proof_2}
\end{align}
where the finiteness is a consequence of $1-\beta > 1-\alpha$. Next, since $\Theta_k>0$
and by definition of $\varphi_k$, it follows that any maximizer for $\Theta_k$ is
contained in the set
\[
  A := \bigl\{(t,x,\hat{t},\hat{x})\in\overline{\S}_T\times\overline{\S}_T
    : u^{\rho}(t,x)-v^{\rho}(\hat{t},\hat{x})\geq 0\bigr\}.
\]
Since $u^{\rho}-v^{\rho}$ is upper semi-continuous and by \eqref{eq:comparison_proof_1}
and~\eqref{eq:comparison_proof_2}, it follows that $A$ is compact. In particular,
for each $k\in\N$ there exists a maximizer $(t_k,x_k,\hat{t}_k,\hat{x}_k)$ for
$\Theta_k$ and, upon dropping to a subsequence if necessary, we may assume that
this sequence converges. Next, by maximality of $(t_k,x_k,\hat{t}_k,\hat{x}_k)$,
definition of $\varphi_k$ and since $\Theta_k\geq 0$, we see that
\begin{align*}
  0 \leq \frac{k}{2}\bigl(|t_{k}-\hat{t}_{k}|^{2} + |x_{k}-\hat{x}_{k}|^{2}\bigr)
  &= u^{\rho}(t_{k},x_{k}) - v^{\rho}(\hat{t}_{k},\hat{x}_{k}) - \Theta_{k}\\
  &\leq \sup_{(t,x,\hat{t},\hat{x})\in A} \bigl\{u^{\rho}(t,x)
    - v^{\rho}(\hat{t},\hat{x})\bigr\}
  < \infty,
\end{align*}
and we conclude that we must in fact have
\[
  (\bar{t},\bar{x})
  := \lim_{k\to\infty}(t_{k},x_{k})
  = \lim_{k\to\infty}(\hat{t}_{k},\hat{x}_{k}).
\]
Moreover, combining $\Theta_k\geq \Theta$ and upper semi-continuity of $u^\rho$
and $-v^\rho$, it follows that
\begin{align*}
  0 &\leq \limsup_{k\to\infty}\frac{k}{2}\bigl(|t_{k}-\hat{t}_{k}|^{2}
    + |x_{k}-\hat{x}_{k}|^{2}\bigr)\\
  &= \limsup_{k\to\infty}\bigl\{u^{\rho}(t_{k},x_{k})
    - v^{\rho}(\hat{t}_{k},\hat{x}_{k}) - \Theta_{k}\bigr\}
  \leq u^{\rho}(\bar{t},\bar{x}) - v^{\rho}(\bar{t},\bar{x}) - \Theta
  \leq 0,
\end{align*}
where the last inequality is due to the definition of $\Theta$. But then, again
using upper semi-continuity of $u^\rho$ and $-v^\rho$, we conclude that
\begin{equation}\label{eq:proof-CP-convergence}
\begin{aligned}
  &\lim_{k\to\infty} u^{\rho}(t_{k},x_{k}) = u^{\rho}(\bar{t},\bar{x}),&
  &\lim_{k\to\infty} v^{\rho}(\hat{t}_{k},\hat{x}_{k}) = v^{\rho}(\bar{t},\bar{x}),\\
  &\lim_{k\to\infty} \Theta_{k} = \Theta = u^{\rho}(\bar{t},\bar{x}) - v^{\rho}(\bar{t},\bar{x}),&
  &\lim_{k\to\infty} \frac{k}{2}\bigl(|t_{k}-\hat{t}_{k}|^{2}
    + |x_{k}-\hat{x}_{k}|^{2}\bigr) = 0.
\end{aligned}
\end{equation}
From this, we furthermore conclude that $\bar{t}<T$ since otherwise
\[
  0 < \Theta = u^\rho(\bar{t},\bar{x}) - v^\rho(\bar{t},\bar{x})
  = u(T,\bar{x}) - v(T,\bar{x})
    + \frac{1}{\rho}\Bigl(u(T,\bar{x}) + v(T,\bar{x}) - 2\psi(T,\bar{x})\Bigr)
  \leq 0,
\]
where we have used that $u(T,\argdot)\leq v(T,\argdot)$ by assumption, $u,v \leq \psi$
by construction, and $v\geq 0$. From this, it follows that we may subsequently assume
that $t_k,\hat{t}_k < T$ for all $k\in\N$ by dropping to another subsequence
if necessary. In particular, we can apply Ishii's lemma \cite[Theorem 3.2]{crandall1992user},
yielding the existence of matrices $M_k,\hat{M}_k\in\cS^{1+N}$ satisfying
\begin{equation}\label{eq:Ishii-Matrix}
  \begin{pmatrix}
    M_{k} & 0\\
	0 & -\hat{M}_{k}
  \end{pmatrix}
  \leq \begin{pmatrix}
    {\mathrm I} & -{\mathrm I}\\
	-{\mathrm I} & {\mathrm I} 
  \end{pmatrix}
\end{equation}
where ${\mathrm I}\in\cS^{1+N}$ is the identity matrix and such that
\begin{align*}
  \bigl(k(t_{k}-\hat{t}_{k}),k(x_{k}-\hat{x}_{k}),M_{k}\bigr)
    &\in\overline{J}^{2,+} u^{\rho}(t_{k},x_{k}),\\
  \bigl(k(t_{k}-\hat{t}_{k}),k(x_{k}-\hat{x}_{k}),\hat{M}_{k}\bigr)
    &\in\overline{J}^{2,-}v^{\rho}(\hat{t}_{k},\hat{x}_{k}),
\end{align*}
where $\overline{J}^{2,+} u^{\rho}(t_{k},x_{k})$ and $\overline{J}^{2,-}v^{\rho}(\hat{t}_{k},\hat{x}_{k})$
denote, as usual, the closures of the second-order super- and subjets, respectively.
By Lemma~\ref{lem:perturbation}, we conclude that
\begin{align}
  - \frac{\bar{\kappa}}{\rho}\label{eq:proof-CP-Step2-contr.-1}
  &\geq \min\Bigl\{F\bigl(x_{k},k(t_{k}-\hat{t}_{k}),k(x_{k}-\hat{x}_{k}),M_{k}\bigr),
    u^{\rho}(t_{k},x_{k}) - \cM[u^{\rho}](t_{k},x_{k})\Bigr\},\\
	\frac{\bar{\kappa}}{\rho}\label{eq:proof-CP-Step2-contr.-2}
  &\leq \min\Bigl\{F\bigl(\hat{x}_{k},k(t_{k}-\hat{t}_{k}),k(x_{k}-\hat{x}_{k}),\hat{M}_{k}\bigr),
    v^{\rho}(\hat{t}_{k},\hat{x}_{k}) - \cM[v^{\rho}]_*(\hat{t}_{k},\hat{x}_{k})\Bigr\},
\end{align}
where $\bar{\kappa} := \inf_{(t,x,\hat{t},\hat{x})\in A} \kappa(t,x) > 0$.

Step 2. Handling of the non-locality. We now show that we can strengthen
the inequality in \eqref{eq:proof-CP-Step2-contr.-1} in the sense that, after
dropping to a subsequence, we in fact have
\begin{equation}\label{eq:proof-CP-step3-wlog}
  - \frac{\bar{\kappa}}{\rho}
  \geq F\bigl(x_{k},k(t_{k}-\hat{t}_{k}),k(x_{k}-\hat{x}_{k}),M_{k}\bigr),
  \quad k\in\N.
\end{equation}
Assume by contradiction that this is not the case, i.e.\ \eqref{eq:proof-CP-step3-wlog}
holds for at most finitely many $k\in\N$. By \eqref{eq:proof-CP-Step2-contr.-1}, this is only
possible if there exists $K_0\in\N$ such that
\[
  u^{\rho}(t_{k},x_{k}) \leq \cM[u^{\rho}](t_{k},x_{k}) - \frac{\bar{\kappa}}{\rho},
  \quad k\geq K_0.
\]
Since $(t_{k},x_{k},\hat{t}_k,\hat{x_k})\in A$ eventually and $v^\rho\geq 0$, we
conclude that $u^\rho(t_k,x_k)\geq 0$. With this and since $\cM[u^{\rho}] = -1$
on $\overline{\S}_T^\emptyset$, it follows that
$(t_k,x_k)\in\overline{\S}_T\setminus\overline{\S}_T^\emptyset$ for all $k\geq K_0$
if $K_0$ is chosen large enough. But then we must also have
$(\bar{t},\bar{x})\in\overline{\S}_T\setminus\overline{\S}_T^\emptyset$
since $\overline{\S}_T^\emptyset$ is open relative to $\overline{\S}_T$.
We proceed to show that
\[
  (\bar{t},\bar{x})\not\in\partial\overline{\S}_T^\emptyset
  := \bigl\{ (t,x)=(t,w,p)\in\overline{\S}_T : w = K(t,0) \bigr\},
\]
i.e.\ $(\bar{t},\bar{x})$ is an interior point of $\overline{\S}_T\setminus\overline{\S}_T^\emptyset$.
For this, let us first use that by~\eqref{eq:proof-CP-convergence} we can make
$K_0$ larger to guarantee that
\[
  \Theta = u^{\rho}(\bar{t},\bar{x}) - v^{\rho}(\bar{t},\bar{x})
  \leq u^{\rho}(t_{k},x_{k}) - v^{\rho}(\hat{t}_{k},\hat{x}_{k}) + \frac{\bar\kappa}{4\rho},
  \quad k\geq K_0.
\]
Similarly, with $K_0$ possibly even larger, upper semi-continuity of $\cM[u^\rho]$
implies that we have
\[
  \cM[u^\rho](t_k,x_k) \leq \cM[u^\rho](\bar t,\bar x) + \frac{\bar\kappa}{4\rho},
  \quad k\geq K_0.
\]
Since $\cD(\bar{t},\bar{x})$ is non-empty and compact, upper semi-continuity of
$u^\rho$ and continuity of $\Gamma$ implies the existence of $q\in\cD(\bar{t},\bar{x})$
such that
\[
  \cM[u^{\rho}](\bar{t},\bar{x}) = \sum_{n=1}^N \bar{p}^n
    \E\Bigl[u^\rho\Bigl(t,\Gamma\bigl(q\mu^n + (1-q)\cN,\bar{t},\bar{x},q\bigr)\Bigr)\Bigr].
\]
Combining the previous estimates, we conclude that
\begin{equation}\label{eq:proof-CP-latter-three}
  \Theta \leq \sum_{n=1}^N \bar{p}^n
    \E\Bigl[u^\rho\Bigl(t,\Gamma\bigl(q\mu^n + (1-q)\cN,\bar{t},\bar{x},q\bigr)\Bigr)\Bigr]
  - v^{\rho}(\hat{t}_{k},\hat{x}_{k}) - \frac{\bar{\kappa}}{2\rho}.
\end{equation}
Now if $(\bar{t},\bar{x})\in\partial\overline{\S}_T^\emptyset$, it follows that
$\cD(\bar{t},\bar{x}) = \{0\}$ and $\bar{w} = K(t,0)$, so that
\[
  \Gamma\bigl(q\mu^n + (1-q)\cN,\bar{t},\bar{x},q\bigr)
  = \Gamma\bigl(\cN,\bar{t},\bar{x},0\bigr) = (0,\bar{p}),
\]
from which we conclude that
\[
  u^\rho\Bigl(t,\Gamma\bigl(q\mu^n + (1-q)\cN,\bar{t},\bar{x},q\bigr)\Bigr)
  = \frac{\rho+1}{\rho}u(\bar{t},0,\bar{p}) - \frac{1}{\rho}\psi(\bar{t},0,\bar{p})
  = - \frac{1}{\rho}\psi(\bar{t},0,\bar{p})
\]
where we have used the boundary condition~\eqref{eq:boundary}.
But then we arrive at the contradiction
\[
  0 \leq \Theta
  \leq -\frac{1}{\rho}\psi(\bar{t},0,\bar{p}) - \frac{\bar{\kappa}}{2\rho} < 0,
\]
and we conclude that $(\bar{t},\bar{x})\not\in\partial\overline{\S}_T^\emptyset$.
In particular, it follows that we can assume that
$(\hat{t}_k,\hat{x}_k)\in\overline{\S}_T\setminus\overline{\S}_T^\emptyset$
for all $k\geq K_0$. First using \eqref{eq:proof-CP-Step2-contr.-2} together with
lower semi-continuity of $\cM[v^\rho]$ on $\overline{\S}_T\setminus\overline{\S}_T^\emptyset$
by Lemma~\ref{lem:regularity-intervention}, we can make $K_0$ larger if necessary
so that
\begin{align*}
  v^\rho(\hat{t}_k,\hat{x}_k)
  \geq \cM[v^\rho(\hat{t}_k,\hat{x}_k)]_* + \frac{\bar\kappa}{\rho}
  &= \cM[v^\rho(\hat{t}_k,\hat{x}_k)] + \frac{\bar\kappa}{\rho}\\
  &\geq \cM[v^\rho(\bar{t},\bar{x})] + \frac{\bar\kappa}{2\rho}\\
  &\geq \sum_{n=1}^N \bar{p}^n
    \E\Bigl[v^\rho\Bigl(t,\Gamma\bigl(q\mu^n + (1-q)\cN,\bar{t},\bar{x},q\bigr)\Bigr)\Bigr]
    + \frac{\bar\kappa}{2\rho},
  & k&\geq K_0,
\end{align*}
where we have used that $q\in\cD(\bar{t},\bar{x})$ to obtain the last
inequality. Combining this with~\eqref{eq:proof-CP-latter-three} and using the definition
of $\Theta$ yields
\[
  \Theta \leq \sum_{n=1}^N \bar{p}^n\E\Bigl[
	u^\rho\Bigl(t,\Gamma\bigl(q\mu^n + (1-q)\cN,\bar{t},\bar{x},q\bigr)\Bigr)
    - v^\rho\Bigl(t,\Gamma\bigl(q\mu^n + (1-q)\cN,\bar{t},\bar{x},q\bigr)\Bigr)
  \Bigr] - \frac{\bar\kappa}{\rho}
  \leq \Theta - \frac{\bar\kappa}{\rho} < \Theta,
\]
which is the desired contradiction. We may therefore subsequently assume that
\eqref{eq:proof-CP-step3-wlog} holds.

Step 3. Conclusion. Combining \eqref{eq:proof-CP-Step2-contr.-2} and
\eqref{eq:proof-CP-step3-wlog} reveals
\begin{align*}
  2\frac{\bar\kappa}{\rho}
  &\leq F\bigl(\hat{x}_{k},k(t_{k}-\hat{t}_{k}),k(x_{k}-\hat{x}_{k}),\hat{M}_{k}\bigr)
    - F\bigl(x_{k},k(t_{k}-\hat{t}_{k}),k(x_{k}-\hat{x}_{k}),M_{k}\bigr)\\
  &= - \max_{\pi\in\Pi} \cH\bigl(\hat{x}_{k},\pi,k(x_{k}-\hat{x}_{k}),\hat{M}_{k}\bigr)
    + \max_{\pi\in\Pi} \cH\bigl(x_{k},\pi,k(x_{k}-\hat{x}_{k}),M_{k}\bigr)\\
  &\leq \max_{\pi\in\Pi}\Bigl[\cH\bigl(x_{k},\pi,k(x_{k}-\hat{x}_{k}),M_{k}\bigr)
    - \cH\bigl(\hat{x}_{k},\pi,k(x_{k}-\hat{x}_{k}),\hat{M}_{k}\bigr)\Bigr].
\end{align*}
By boundedness of $\Pi$ and $\Delta^N$, it follows that the state coefficient
functions $x\mapsto f(x,\pi)$ and $x\mapsto \Sigma(x,\pi)$ are Lipschitz continuous
in $x$, uniformly in $\pi$. Combining this with~\eqref{eq:Ishii-Matrix}, it follows
that there exists a constant $C>0$ such that
\[
  0 < 2\frac{\bar\kappa}{\rho}
  \leq \max_{\pi\in\Pi}\Bigl[\cH\bigl(x_{k},\pi,k(x_{k}-\hat{x}_{k}),M_{k}\bigr)
    - \cH\bigl(\hat{x}_{k},\pi,k(x_{k}-\hat{x}_{k}),\hat{M}_{k}\bigr)\Bigr]
  \leq C\frac{k}{2}|x_k-\hat{x}_k|^2,\quad k\in\N,
\]
and we arrive at the final contradiction as the right-hand side tends to zero
for $k\to\infty $ by~\eqref{eq:proof-CP-convergence}.\fakeQED

\paragraph{Proof of Theorem~\ref{th:V+-is-viscosity-subsolution}.}

In light of the discussion preceding Theorem~\ref{th:V+-is-viscosity-subsolution}
we only have to verify the viscosity subsolution property. We begin by observing
that $V^+$ is upper semi-continuous as the pointwise minimum of upper semi-continuous
functions. Hence $(V^+)^* = V^+$ and $\cM[V^+]^* = \cM[V^+]$ by
Lemma~\ref{lem:regularity-intervention}. Towards the subsolution property, we
argue by contradiction and assume that there exists a test function $\varphi\in\cC^{1,2}(\S_T)$
and $(\bar t,\bar x)\in\S_T$ such that $V^+-\varphi$ attains a strict global maximum
equal to zero at $(\bar t,\bar x)$ and
\[
  \min\Bigl\{
    F\bigl(\bar x, \varphi_t(\bar t,\bar x),\rD_x\varphi(\bar t,\bar x),\rD_x^2\varphi(\bar t,\bar x)\bigr),
	V^+(\bar t,\bar x) - \cM[V^+](\bar t,\bar x)
  \Bigr\} = 2\kappa > 0
\]
for some $\kappa>0$.
For $\delta>0$, we define
\[
  \cB_{\delta}(\bar{t},\bar{x})
  := \bigl\{(t,x)\in\overline{\S}_T : |(t,x)-(\bar{t},\bar{x})|<\delta\bigr\}
\]
and denote by $\overline{\cB}_{\delta}(\bar{t},\bar{x})$ its closure in
$\overline{\S}_T$. By continuity of $F$ (which follows from compactness of $\Pi$),
$V^+=\varphi$ in $(\bar t,\bar x)$, and lower semi-continuity of $\varphi - \cM[V^+]$,
there exists $\delta>0$ small enough such that $t+\delta < T$ and
\begin{equation}\label{eq:proof-visco-sub-ineq-on-ball}
  \min\Bigl\{
    F\bigl(x, \varphi_t(t,x),\rD_x\varphi(t,x),\rD_x^2\varphi(t,x)\bigr),
    \varphi(t,x) - \cM[V^+](t,x)
  \Bigr\} \geq \kappa,
  \quad (t,x)\in\overline{\cB}_{\delta}(\bar t,\bar x).
\end{equation}
The compactness of $\overline{\cB}_{\delta}(\bar{t},\bar{x})\setminus\cB_{\delta/2}(\bar{t},\bar{x})$
and the upper semi-continuity of $V^{+}-\varphi$ together with the strictness of
the maximum at $(\bar{t},\bar{x})$ imply the existence of $\eta_0\in(0,\kappa)$ satisfying
\begin{equation}\label{eq:h_is_v_outside_half_ball}
  V^{+}(t,x) + \eta_0 \leq \varphi(t,x),
  \quad (t,x)\in\overline{\cB}_{\delta}(\bar{t},\bar{x})
	\setminus\cB_{\delta/2}(\bar{t},\bar{x}).
\end{equation}
With this, we define for $\eta\in(0,\eta_0)$ the function $\varphi^{\eta}
:= \varphi-\eta$ and
\[
  h^{\eta} := \begin{cases}
    V^{+} \wedge \varphi^{\eta} &\text{on }
	  \overline{\cB}_{\delta}(\bar{t},\bar{x}),\\
    V^{+} &\text{else}.
  \end{cases}
\]
Since
\[
  \varphi^{\eta}(\bar{t},\bar{x})
  = \varphi(\bar{t},\bar{x}) - \eta
  = V^{+}(\bar{t},\bar{x}) - \eta
  < V^{+}(\bar{t},\bar{x}),
\]
it follows that
\[
  h^\eta(\bar t,\bar x) = \varphi^{\eta}(\bar{t},\bar{x}) < V^{+}(\bar{t},\bar{x}),
\]
so that we obtain a contradiction to the minimality of $V^+$ if we can show that
$h^\eta\in\cV^+$. Since both $V^+\in\USC$ and $V^+\wedge\varphi^\eta\in\USC$
and $h^\eta = V^\eta$ outside of $\cB_{\delta/2}(\bar t,\bar x)$
by~\eqref{eq:h_is_v_outside_half_ball} and since $\eta<\eta_0$, we conclude that
$h^\eta\in\USC$, i.e.\ $h^\eta$ satisfies $(\cV^+_1)$. Next, since $V^+$ and $h^\eta$
differ at most on a compact set, it follows that $h^\eta$ satisfies the lower
boundedness and growth condition $(\cV^+_2)$. Finally, since $\bar t+\delta < T$,
we conclude that $h^\eta(T,\argdot) = V^+(T,\argdot)$, so that $h^\eta$ is seen
to satisfy the terminal inequality $(\cV^+_3)$ and we are left with verifying
$(\cV^+_4)$ and $(\cV^+_5)$.

Let us start by establishing $(\cV^+_5)$ for those
$(t,x)\in\overline{S}_T$ for which $h^\eta(t,x) = \varphi^\eta(t,x)$. By
\eqref{eq:h_is_v_outside_half_ball}, this is only possible if
$(t,x)\in\cB_{\delta/2}(\bar t,\bar x)$. Hence, using $h^{\eta}(t,x)
= \varphi(t,x) - \eta$ and $h^{\eta}\leq V^{+}$ for the first inequality followed
by \eqref{eq:proof-visco-sub-ineq-on-ball} for the second inequality yields
\[
  h^{\eta}(t,x) - \cM[h^{\eta}](t,x)
  \geq \varphi(t,x) - \eta + \cM[V^{+}](t,x)
  \geq \kappa-\eta > 0
\]
as claimed. On the other hand, if $(t,x)$ is such that $h^\eta(t,x) = V^+(t,x)$,
we can use that $h^\eta\leq V^+$ implies $\cM[h^\eta]\leq \cM[V^+]$ and that $V^+$
satisfies $(\cV^+_5)$ to conclude that
\[
  h^{\eta}(t,x) - \cM[h^\eta](t,x)
  \geq V^+(t,x) - \cM[V^+](t,x) \geq 0.
\]
Thus $h^\eta$ satisfies $(\cV^+_5)$.

Let us now turn to (V4) and fix $(t,x)\in\overline{\S}_T$, $u=(\pi,\nu)\in\cA(x)$,
$k\in\N$, a pair of $\cY^u$-stopping times $\theta\leq\rho$ taking values in
$[\tilde\tau_k,\tilde\tau_{k+1}]\cap[t,T]$, and a $\cY^u_\theta$-measurable and
$\S$-valued random variable $\xi$ with $\E[|\xi|^2]<\infty$. We introduce the
$\cY^u_\theta$-measurable event
\[
  A:=\bigl\{(\theta,\xi)\in\cB_{\delta/2}(\bar{t},\bar{x})\text{ and }
    \varphi^{\eta}(\theta,\xi)<V^{+}(\theta,\xi)\bigr\}
\]
and the $\cY^u$-stopping time
\[
  \vartheta := \inf\bigl\{s\in[\theta,T] : (s,X^{u;\theta,\xi}_s)
    \not\in\cB_{\delta/2}(\bar t,\bar x)\bigr\}.
\]
Note that $\vartheta<T$ since $\{\varphi^\eta < V^+\}
\subseteq\cB_{\delta/2}(\bar t,\bar x)$ and $\vartheta = \theta$ on the complement
of $A$. On the other hand, on $A$, we may apply It\^{o}'s formula to obtain
\begin{align*}
  \ind_{A}h^\eta(\theta,\xi)
  &= \ind_{A}\varphi^\eta(\theta,\xi)\\
  &= \ind_{A}\varphi^\eta\bigl(\vartheta\wedge\rho,
    X^{u;\theta,\xi}_{(\vartheta\wedge\rho)-}\bigr)\\
    &\quad - \ind_{A}\int_\theta^{\vartheta\wedge\rho} \varphi^\eta_t(s,X^{u;\theta,\xi}_s)
	  + \cH\bigl(X^{u;\theta,\xi}_s,\pi_s,\rD_x\varphi^\eta_t(s,X^{u;\theta,\xi}_s),
	    \rD_x^2\varphi^\eta_t(s,X^{u;\theta,\xi}_s)\bigr)\rd s\\
	&\hspace{6.5cm} - \ind_{A}\int_\theta^{\vartheta\wedge\rho} \Sigma(X^{u;\theta,\xi}_s,\pi_s)^\intercal
      \rD_x\varphi^\eta_t(s,X^{u;\theta,\xi}_s) \rd I^\nu_s.
\end{align*}
The integral with respect to $I^\nu$ is a martingale since the integrand is bounded
on $A$ by definition of stopping time $\vartheta$. Moreover, the partial derivatives
of $\varphi$ and $\varphi^\eta$ coincide, so that \eqref{eq:proof-visco-sub-ineq-on-ball}
implies, on $A$,
\begin{multline*}
  - \varphi^\eta_t(s,X^{u;\theta,\xi}_s)
  - \cH\bigl(X^{u;\theta,\xi}_s,\pi_s,\rD_x\varphi^\eta_t(s,X^{u;\theta,\xi}_s)
    \rD_x^2\varphi^\eta_t(s,X^{u;\theta,\xi}_s)\bigr)\\
  \geq F\bigl(X^{u;\theta,\xi}_s,\varphi^\eta_t(s,X^{u;\theta,\xi}_s),
    \rD_x\varphi^\eta_t(s,X^{u;\theta,\xi}_s), \rD_x^2\varphi^\eta_t(s,X^{u;\theta,\xi}_s)\bigr)
  > 0,\qquad s\in[\theta,\vartheta\wedge\rho).
\end{multline*}
Finally, we have $\varphi^\eta\geq h^\eta$ on $\cB_{\delta/2}(\bar t,\bar x)$ and
hence, upon combining these arguments, we conclude that
\begin{equation}\label{eq:h_satisfies_v4_1}
  \ind_{A}h^\eta(\theta,\xi)
  \geq \E\Bigl[\ind_{A}\varphi^\eta\bigl(\vartheta\wedge\rho,
    X^{u;\theta,\xi}_{(\vartheta\wedge\rho)-}\bigr)\Big|\cY^{\nu}_\theta\Bigr]
  \geq \E\Bigl[\ind_{A}h^\eta\bigl(\vartheta\wedge\rho,
  X^{u;\theta,\xi}_{(\vartheta\wedge\rho)-}\bigr)\Big|\cY^{\nu}_\theta\Bigr].
\end{equation}
Next, let $B:=\{\vartheta<\rho\}$. Since $\rho\leq\tilde\tau_{k+1}$, it follows
that
\[
  X^{u;\theta,\xi}_{(\vartheta\wedge\rho)-} = X^{u;\theta,\xi}_{\vartheta}
  \in\partial\cB_{\delta/2}(\bar t,\bar x)
  := \bigl\{(t,x)\in\overline{\S}_T : |(t,x)-(\bar t,\bar x)| = \delta/2\bigr\}
\]
since $X$ is continuous on $[\tilde\tau_{k},\tilde\tau_{k+1})\cap[0,T]$. But this
implies
\begin{align}
  \ind_{A}\ind_{B} h^\eta\bigl(s,X^{u;\theta,\xi}_{(\vartheta\wedge\rho)-}\bigr)
  = \ind_{A}\ind_{B} V^+\bigl(s,X^{u;\theta,\xi}_{\vartheta}\bigr)
  &\geq \ind_{A}\ind_{B} \E\Bigl[V^+\bigl(s,X^{u;\theta,\xi}_{\rho-}\bigr)\Big|\cY^\nu_\vartheta\Bigr]\notag\\
  &\geq \ind_{A}\ind_{B} \E\Bigl[h^\eta\bigl(s,X^{u;\theta,\xi}_{\rho-}\bigr)\Big|\cY^\nu_\vartheta\Bigr]
  \label{eq:h_satisfies_v4_2}
\end{align}
by property $(\cV^+_4)$ of $V^+$, pathwise uniqueness of $X$, and the general inequality
$h^\eta\leq V^+$. Finally, on the complement $A^c$ of $A$, we have
$h^\eta(\theta,\xi)=V^+(\theta,\xi)$, and hence the same argument yields
\begin{equation}\label{eq:h_satisfies_v4_3}
  \ind_{A^c}h^\eta(\theta,\xi)
  \geq \ind_{A^c}\E\Bigl[h^\eta\bigl(s,X^{u;\theta,\xi}_{\rho-}\bigr)\Big|\cY^\nu_\theta\Bigr].
\end{equation}
Combining \eqref{eq:h_satisfies_v4_1}, \eqref{eq:h_satisfies_v4_2}, and
\eqref{eq:h_satisfies_v4_3} therefore shows that
\[
  h^\eta(\theta,\xi)
  \geq \E\Bigl[h^\eta\bigl(s,X^{u;\theta,\xi}_{\rho-}\bigr)\Big|\cY^\nu_\theta\Bigr],
\]
that is, $h^\eta$ satisfies $(\cV^+_4)$ and hence the proof is complete.\fakeQED

\paragraph{Proof of Proposition~\ref{th:V+-is-visco-super}.}

The growth condition on $(V^+)_*$ is obvious since it is already satisfied for $V^+$.
Regarding the terminal condition, note that $V^+(T,x) = U(w)$ implies that
$V^+(T,x)_* \leq U(w)$, so we only have to prove the reverse inequality. For this
let $\{(t_k,x_k)\}_{k\in\N}\subset\overline{S}_T$ be a sequence converging to $(T,x)$
with $x\in\S$ such that $V^+(t_k,x_k)\to V^+(T,x)_*$. Consider the trading strategy
$\pi = 0\in\Pi$ and an expert opinion strategy $\nu=\circ$ which does not purchase any
expert opinions at all. With this, it follows that $W^{0,\circ}$ is constant.
Using the property $(\cV^+_4)$ of $V^+$ and $V^+(T,x) = U(w)$, it follows that
\[
  (V^+)_*(T,x)
  = \lim_{k\to\infty} V^+(t_k,x_k)
  \geq \limsup_{k\to\infty} \E\Bigl[V^+(T,X^{\pi,\circ;t_k,x_k})\Big|\cY^\circ_{t_k}\Bigr]
  = \limsup_{k\to\infty} U(w_k) = U(w),
\]
where we use the notation $x_k=(w_k,p_k)$ for $k\in\N$. In total, we have therefore
argued that $V^+(T,x)_* = U(w)$ as claimed and we can move on to the viscosity
supersolution property.

For this, let $h$ be a measurable function satisfying $(\cV^+_2)$ to $(\cV^+_5)$. We
first observe that
\[
  h \geq \cM[h] \geq \cM[h_*] \geq \cM[h_*]_*
\]
since $h$ satisfies $(\cV^+_5)$. Now $\cM[h_*]_*$ is a lower semi-continuous function,
hence $h_*\geq \cM[h_*]_*$ since $h_*$ is by definition the largest lower semi-continuous
function dominated by $h$. Now fix $(\bar t,\bar x)\in\S_T$ and let
$\varphi\in\cC^{1,2}(\S_T)$ be a test function for the viscosity supersolution property
of $h$ at $(\bar t,\bar x)$, which is to say that $h_*(\bar t,\bar x) = \varphi(\bar t,\bar x)$
and $\varphi - h_*$ attains a global maximum at $(\bar t,\bar x)$. To conclude,
it suffices to show that
\begin{equation}\label{eq:supersol_to_show}
  F\bigl(\bar x,\varphi_t(\bar t,\bar x),\rD_x\varphi(\bar t,\bar x),
  \rD^2_x\varphi(\bar t,\bar x)\bigr) \geq 0.
\end{equation}
For this, let us fix a sequence $\{(t_k,x_k)\}_{k\in\N}\subset\S_T$ converging to
$(\bar t,\bar x)$ such that $h(t_k,x_k)\to h_*(\bar t,\bar x)$. By continuity of
$\varphi$, it follows that
\[
  0 \leq \gamma_k := h(t_k,x_k) - \varphi(t_k,x_k) \to 0\quad\text{as }k\to\infty.
\]
Now choose a sequence $\{\delta_k\}_{k\in\N}$ of strictly positive real numbers
such that
\[
  \lim_{k\to\infty} \delta_k = 0 = \lim_{k\to\infty} \frac{\gamma_k}{\delta_k}.
\]
Let moreover $\eta>0$, consider a constant trading strategy $\pi\in\Pi$, and
denote again by $\circ$ the expert opinion strategy which does not purchase any
expert opinions at all. Writing $X^k := X^{\pi,\circ;t_k,x_k}$, we
introduce the stopping times
\[
  \rho_k := \inf\bigl\{s\in[t_k,T] : |X^k_s - x_k| > \eta\bigr\}
    \wedge (t_k + \delta_k) \wedge T,
  \quad k\in\N.
\]
Using property $(\cV^+_4)$, $h\geq \varphi$, and finally It\^{o}'s formula, it
follows that
\begin{align*}
  h(t_k,x_k)
  &\geq \E\bigl[h(\rho_k,X^k_{\rho_k})\bigr]\\
  &\geq \E\bigl[\varphi(\rho_k,X^k_{\rho_k})\bigr]\\
  &= \varphi(t_k,x_k) + \E\Bigl[\int_{t_k}^{\rho_k}
    \varphi_t(s,X^k_s) + \cH\bigl(X^k_s,\pi,
	  \rD_x\varphi(s,X^k_s), \rD^2_x\varphi(s,X^k_s)\bigr)
  \rd s\Bigr],
\end{align*}
where the stochastic integral vanishes since its integrand is bounded on $[t_k,\rho_k]$.
Rearranging terms and dividing by $\delta_k$ shows that
\[
  \frac{\gamma_k}{\delta_k} - \E\Bigl[\frac{1}{\delta_k}\int_{t_k}^{\rho_k}
    \varphi_t(s,X^k_s) + \cH\bigl(X^k_s,\pi,
	  \rD_x\varphi(s,X^k_s), \rD^2_x\varphi(s,X^k_s)\bigr)
  \rd s \Bigr] \geq 0.
\]
Now $\rho_k(\omega) = t_k + \delta_k$ for eventually all $k\in\N$ and $\P$-almost
every $\omega\in\Omega$. Hence, as $k\to\infty$, the mean-value theorem implies
\[
  - \varphi_t(\bar t,\bar x) - \cH\bigl(\bar x,\pi,
	\rD_x\varphi(\bar t,\bar x), \rD^2_x\varphi(\bar t,\bar x)\bigr) \geq 0.
\]
Since this holds for any $\pi\in\Pi$, we conclude that \eqref{eq:supersol_to_show}
holds and the proof is complete.\fakeQED

\paragraph{Proof of Theorem~\ref{th:V--is-viscosity-supersolution}.}

We begin with the growth condition on $V^-$. Since we already know that
$V^-\leq V^\cont$ and that $V^+$ satisfies the desired growth condition,
we only have to show that $V^\cont\leq V^+$. For this, we fix
$(t,x)=(t,w,p)\in\overline{\cS}\subset\overline{\S}_T$ and consider an
arbitrary strategy $\pi\in\cA^\circ$. Recall that, in particular,
$\nu:=(\pi,\circ)\in\cA(w)$. Hence, we may apply the supermartingale
property $(\cV^+_4)$ of $V^+$ with $(\theta,\rho,\xi)
:=(t,\tau^\pi_\cS,x)$ to obtain
\[
  V^+(t,x) \geq \E\bigl[V^+\bigl(X^{\pi;t,x}_{\tau_\cS^\pi}\bigr)\bigr].
\]
As this holds for any $\pi\in\cA^\circ$, we conclude that
\[
  V^+(t,x)
  \geq \sup_{\pi\in\cA^\circ}\E\bigl[V^+\bigl(X^{\pi;t,x}_{\tau_\cS^\pi}\bigr)\bigr]
  = V^\cont(t,x)
\]
as claimed. The remainder of the proof is very similar in spirit to that of
Theorem~\ref{th:V+-is-viscosity-subsolution}. One of the main differences is that,
in addition to the viscosity supersolution property on $\S_T$, we also have to
establish the boundary inequality $V^-(t,x) \geq V^+(t,x)$ for all
$(t,x)\in\overline{\cS}$. We therefore proceed in two steps.

Step 1. The supersolution property on $\cS$. We first take note
that $V^-$ is lower semi-continuous as the supremum of lower semi-continuous
functions and proceed to argue towards a contradiction by assuming the existence
of $(\bar{t},\bar{x})\in\cS$, a test function $\varphi\in\cC^{1,2}(\cS)$ such that
$V^--\varphi$ has a strict global minimum equal to zero at $(\bar{t},\bar{x})$,
and such that the supersolution property fails, that is
\[
  F\bigl(\bar{x},\varphi_t(\bar{t},\bar{x}),
    \rD_x\varphi(\bar{t},\bar{x}),\rD_x^2\varphi(\bar{t},\bar{x})\bigr)
    = -3\kappa < 0
\]
for some $\kappa>0$. In particular, there exist $\bar{\pi}\in\Pi$ and such that
\[
  -\varphi(\bar{t},\bar{x}) - \cH\bigl(\bar{x},\bar{\pi},
    \rD_x\varphi(\bar{t},\bar{x}),\rD^2_x\varphi(\bar{t},\bar{x})\bigr)
  \leq -2\kappa.
\]
Using the continuity of
$\cH$, it follows that there exists $\delta>0$ sufficiently small to guarantee
that $\overline{\cB}_\delta(\bar{t},\bar{x})\subset\cS$ and such that
\[
  -\varphi_t(t,x) - \cH\bigl(x,\bar{\pi},\rD_x\varphi(t,x),
    \rD^2_x\varphi(t,x)\bigr) \leq -\kappa,
  \quad (t,x)\in\overline{\cB}_\delta(\bar{t},\bar{x}).
\]
Next, since the minimum of $V^--\varphi$ at $(\bar{t},\bar{x})$ is strict and
$\overline{\cB}_\delta(\bar{t},\bar{x})\setminus\cB_{\delta/2}(\bar{t},\bar{x})$
is compact, there exists $\eta_0\in(0,\kappa)$ such that
\begin{equation}\label{eq:proof_supersol_strict_inequality}
  V^-(t,x) + \eta_0 \geq \varphi(t,x),
  \quad (t,x)\in\overline{\cB}_\delta(\bar{t},\bar{x})
    \setminus\cB_{\delta/2}(\bar{t},\bar{x}).
\end{equation}
For $\eta\in(0,\eta_0)$, we finally define $\varphi^\eta = \varphi + \eta$ and
\[
  h^\eta := \begin{cases}
    V^-\vee \varphi^\eta &\text{on }\overline{\cB}_\delta(\bar{t},\bar{x}),\\
    V^- &\text{else.}
  \end{cases}
\]
Observing that
\[
  \varphi^\eta(\bar{t},\bar{x})
  = \varphi(\bar{t},\bar{x}) + \eta
  = V^-(\bar{t},\bar{x}) + \eta
  > V^-(\bar{t},\bar{x}),
\]
we conclude that
\[
  h^\eta(\bar{t},\bar{x})
  = \varphi^\eta(\bar{t},\bar{x})
  > V^-(\bar{t},\bar{x}),
\]
so that we arrive at a contradiction to the maximality property of $V^-$ if
we can show that $h^\eta\in\cV^-$. Since $h^\eta = V^-$ outside of
$\cB_{\delta/2}(\bar{t},\bar{x})$ by \eqref{eq:proof_supersol_strict_inequality}
and since both $V^-$ and $V^-\vee\varphi^\eta$ are lower semi-continuous, it
follows that $h^\eta$ is lower semi-continuous, i.e.\ it satisfies $(\cV^-_1)$.
Lower boundedness and the growth condition in $(\cV^-_2)$ for $h^\eta$ are also
immediate since $0\leq V^-\leq V^+$, $V^+$ satisfies the growth condition,
and $V^-$ and $h^\eta$ differ at most on a compact set. The boundary inequality
$(\cV^-_3)$ holds as $\overline{\cB}_\delta(\bar{t},\bar{x})\subset\cS$,
hence $h^\eta = V^-$ on $\partial^*\cS$, and $V^-$ satisfies the boundary
inequality since $V^-\in\cV^-$.

It remains to show that $h^\eta$ satisfies $(\cV^-_4)$. For this, a stopping
time $\theta$ with values in $[0,T]$ and a square-integrable,
$\cY_\theta$-measurable random variable $\xi$. We denote by
$\pi^\theta\in\cA^\circ$ the strategy obtained from applying $(\cV^-_4)$ to
$V^-$ with initial datum $(\theta,\xi)$. Moreover, we introduce the event
\[
  A := \bigl\{(\theta,\xi)\in\cB_{\delta/2}(\bar{t},\bar{x})
    \text{ and }V^-(\theta,\xi) < \varphi^\eta(\theta,\xi)\bigr\}
\]
and, with this, define
\[
  \hat{\pi} := \bar{\pi}\ind_{A} + \pi^\vartheta\ind_{A^c}
  \quad\text{and}\quad
  \hat{\xi} := X^{\hat{\pi};\theta,\xi}_{\hat{\theta}},
  \quad\text{where}\quad
  \hat{\theta} := \bigl\{s\in[\theta,T] :
    \bigl(s,X_s^{\hat{\pi};\theta,\xi}\bigr)
    \in\partial\cB_{\delta/2}(\bar{t},\bar{x})\bigr\}.
\]
Finally, we let $\pi^{\hat{\theta}}\in\cA^\circ$ denote the strategy obtained
from applying $(\cV^-_4)$ to $V^-$ with initial datum $(\hat{\theta},\hat{\xi})$
and define yet another strategy
\[
  \pi := \hat{\pi}\ind_{[0,\hat{\theta}]}
    + \pi^{\hat{\theta}}\ind_{(\hat{\theta},T]}.
\]
Clearly, we have $\pi\in\cA^\circ$. At this point, we can argue as in
the proof of Theorem~\ref{th:V+-is-viscosity-subsolution} to show that
for any stopping time $\rho$ with values in $[\theta,\tau^\pi_\cS]$
it holds that
\[
  h^\eta(\theta,\xi)
    \leq \E\bigl[h\bigl(\rho,X^{\pi;\theta,\xi}_\rho\bigr)\big|\cY_\theta\bigr],
\]
hence resulting in the desired contradiction. This establishes the viscosity
property on $\cS$.

Step 2. The boundary condition. Since $V^-\in\cV^-$, we already know that
$V^-\leq V^+$ on $\partial^*\cS$ by property $(\cV^-_3)$. Hence it only
remains to show the reverse inequality. Our argument follows the ideas of
\cite[Proposition 5.5]{belak2019utility}.  By contradiction, suppose that
we can find $(\bar{t},\bar{x})\in\partial^*\cS$ such that
\[
  V^-(\bar{t},\bar{x}) - V^+(\bar{t},\bar{x}) = -\kappa < 0
\]
for some $\kappa>0$. Given $\delta,\epsilon>0$, we subsequently write
\begin{align*}
  \cB_{\delta,\epsilon} &:= \bigl\{(t,x)\in\overline{\cS}
    : |\bar{t}-t|<\delta, |\bar{x}-x|<\epsilon\bigr\},\\
  \overline{\cB}_{\delta,\epsilon} &:= \bigl\{(t,x)\in\overline{\cS}
    : |\bar{t}-t|\leq \delta, |\bar{x}-x|\leq\epsilon\bigr\},\\
  \cD_{\delta,\epsilon} &:= \overline{\cB}_{\delta,\epsilon}
    \setminus \cB_{\delta/2,\epsilon/2}.
\end{align*}
Using continuity of $V^+$, there exists $0<\epsilon<\kappa$ such that
\[
  V^-(\bar{t},\bar{x}) - V^+(t,x) \leq -\epsilon < 0,
  \quad (t,x)\in\overline{\cB}_{\epsilon,\epsilon}.
\]
Moreover, by lower semi-continuity, $V^-$ is lower bounded on
$\overline{\cB}_{\epsilon,\epsilon}$ and hence there exists $\beta>0$ sufficiently
small such that, for all $0<\delta\leq\epsilon$ simultaneously,
\begin{equation}\label{eq:exit_viscosity_proof}
  V^-(\bar{t},\bar{x}) < \frac{\epsilon^2}{4\beta} - \epsilon
    + \inf_{(t,x)\in\overline{\cB}_{\delta,\epsilon}} V^-(t,x).
\end{equation}
Now for a constant $K>\delta/(4\beta)$, observe that the function
\[
  \psi(t,x) := V^-(\bar{t},\bar{x}) - \frac{1}{\beta}|x-\bar{x}|^2 - K(\bar{t}-t)
\]
satisfies $\psi_t(t,x) = K$ and $\cH(x,\pi,\rD_x\psi(t,x),\rD^2_x\psi(t,x))$ is
bounded on $\Pi\times\overline{\cB}_{\epsilon,\epsilon}$. Hence, if $K$ is
sufficiently large, it holds that
\[
  F\bigl(x,\psi_t(t,x),\rD_x\psi(t,x),\rD^2_x\psi(t,x)\bigr) < 0,
  \quad (t,x)\in\overline{\cB}_{\epsilon,\epsilon}.
\]
Now with $K$ being fixed, we choose $\delta\leq\min\{\epsilon/(2K),\epsilon\}$.
Since $|\bar{x}-x|\geq \epsilon/2$ and $\bar{t}-t \geq -\delta$ on $\cD_{\delta,\epsilon}$
and by using \eqref{eq:exit_viscosity_proof}, we conclude that
\[
  \psi(t,x)
  = V^-(\bar{t},\bar{x}) - \frac{1}{\beta}|x-\bar{x}|^2 - K(\bar{t}-t)\\
  < V^-(t,x) - \epsilon + K\delta \leq V^-(t,x) - \frac{\epsilon}{2},
  \quad (t,x)\in\cD_{\delta,\epsilon},
\]
where the last inequality follows from $\delta\leq \epsilon/(2K)$. Similarly,
using again that $\bar{t}-t\geq - \delta \geq -\epsilon/(2K)$ on
$\overline{\cB}_{\delta,\epsilon}$, we have
\[
  \psi(t,x)
  \leq V^-(\bar{t},\bar{x}) - K(\bar{t}-t)
  \leq V^-(\bar{t},\bar{x}) + \frac{\epsilon}{2}
  \leq V^+(t,x) - \frac{\epsilon}{2},
  \quad (t,x)\in\overline{\cB}_{\delta,\epsilon}.
\]
Now for $0<\eta<\epsilon/2$, we first define $\psi^\eta := \psi + \eta$ and,
with this,
\[
  h^\eta := \begin{cases}
    V^-\vee \psi^\eta &\text{on }\overline{\cB}_{\delta,\epsilon},\\
    V^- &\text{else}.
  \end{cases}
\]
An argument as in the proof of Theorem~\ref{th:V+-is-viscosity-subsolution}
shows that $h^\eta\in\cV^-$. On the other hand, by definition of $\psi^\eta$
and $\psi$, it holds that
\[
  \psi^\eta(\bar{t},\bar{x})
  = \psi(\bar{t},\bar{x}) + \eta
  > V^-(\bar{t},\bar{x}),
\]
contradicting the maximality of $V^-$ in $\cV^-$, hence concluding the
proof.\fakeQED

\bibliographystyle{plain}
\bibliography{preprint}

\end{document}